%% file: BraggPeakArxivV16.tex
\documentclass[11pt, leqno]{article}

\input{preamble}

\newtheorem{theorem}{\bf Theorem}[section]

\newtheorem{definition}{\bf Definition}[section]
\newtheorem{lemma}{\bf Lemma}[section]
\newtheorem{assumption}{\bf Assumption}[section]
\newtheorem{remark}{\bf Remark}[section]

\usepackage{ulem}

\title{Jump stochastic differential equations for the characterisation of the Bragg peak in proton beam radiotherapy}

\author{
Alastair Crossley\thanks{
Department of Statistics
University of Warwick
Coventry
CV4 7AL, UK. E-mail: \texttt{\{alastair.crossley\}, \{karen.habermann\},  \{emma.horton\}, \{andreas.kyprianou\}@warwick.ac.uk}
},
Karen Habermann$^*$,
Emma Horton$^*$,\\
Jere Koskela\thanks{School of Mathematics, Statistics and Physics
    Newcastle University
    Newcastle upon Tyne
    NE1 7RU, UK. E-mail: \texttt{jere.koskela@newcastle.ac.uk}},
Andreas E. Kyprianou$^*$,
Sarah Osman\thanks{Senior SABR and Imaging Trials and RTTQA Physicist, ULCH Proton Beam facility, 235 Euston Rd., London NW1 2BU, UK. Email \texttt{sarah.osman8@nhs.net}}.
}

\begin{document}

\maketitle
\begin{abstract}
Proton beam radiotherapy stands at the forefront of precision cancer treatment, leveraging the unique physical interactions of proton beams with human tissue to deliver minimal dose upon entry and deposit the therapeutic dose precisely at the so-called Bragg peak, with no residual dose beyond this point.
The Bragg peak is the characteristic maximum that occurs when plotting  the curve 
describing the rate of energy deposition along the length of the proton beam. Moreover, as a natural phenomenon, it is caused by an increase in the rate of nuclear interactions of  protons as their energy  decreases.                                                                                                     From an analytical perspective, Bortfeld \cite{bortfeld} proposed a parametric family of curves that can be accurately calibrated to data replicating the Bragg peak in one dimension.\newline
 We build, from first principles, the very first mathematical model describing the energy deposition of protons. Our approach uses  stochastic differential equations and affords us the luxury of defining the natural analogue of the Bragg curve in two or three dimensions. This work is purely theoretical and provides a new  mathematical framework which is capable of encompassing models built using  Geant4 Monte Carlo, at one extreme, to pencil beam calculations with Bortfeld curves at the other.
 \medskip
 
 \end{abstract}





\section{The physics of proton beam therapy} 

The fundamental basis of proton beam therapy pertains to the way high energy protons decelerate through interactions with subatomic particles in the surrounding tissue. This is due to the fact that deceleration is synonymous with energy deposition. As the energy of a proton decreases, its rate of interaction with subatomic particles increases, leading to a greater rate of energy deposition. This means that, to a certain extent, the rate and location at which the majority of the energy from a proton beam is deposited is controllable, causing minimal exposure to surrounding healthy tissues. It is this phenomenon that distinguishes proton beam therapy from photon therapy, which uses high energy photons.
 
In this article we use the basic principles of atomic physics to build a mathematical model for the dynamics of protons in a proton beam travelling through matter in three-dimensional space. In particular, we develop a jump stochastic differential equation to model the trajectories of protons and their associated energies. This affords us the luxury of being able to demonstrate the existence of a function, in one, two or three dimensions, which describes the energy deposition of proton beams, extending the notion of the so-called {\it Bragg curve} to higher dimensions. For illustration purposes only, we show how, given model choices, one can reconstruct the aforementioned function via Monte Carlo simulations.
Our approach also gives a completely new mathematical framework into which we can embed most existing treatments of proton beam modelling in the public domain, whether it be high fidelity Monte Carlo simulations using Geant4 or pencil beam calibrations of the Bragg curve using Bortfeld's approximation.

We  start with a brief introduction of the physics of proton beams and some basic modelling choices we will make around it. For a more detailed read, the reader is referred to (among many possible sources) \cite{NZ}, \cite{Harald} and \cite{Gott}.

\subsection{Proton interactions} In order to describe the evolution of protons, we need to introduce the  spatio-directional-energy phase space in which they are described. To that end, let $D\subset
\mathbb R^3$ denote a closed and bounded convex spatial domain, $\mathbb S_2$ be
the unit sphere in $\mathbb R^3$ and $\mathcal{E} = [\mathtt{e}_\mathtt{min},
\mathtt{e}_\mathtt{max}]\subset[0,\infty)$ be the range of energies a proton can take. We thus define
the energy-position-direction phase space as
$
  \mathcal{C} = \mathcal{E}\times D\times \mathbb{S}_2 .
$

\begin{wrapfigure}{l}{6cm}
    \vspace{-15pt}
    \begin{tikzpicture}[scale=0.5]
        \nucleus
        \electron{1.5}{0.75}{80}
        \electron{1.2}{1.4}{260}
        \electron{4}{2}{30}
        \electron{4}{3}{180}
        \protoncollision{-6.}{0.}{160}
        \inelastic{-6.}{-2.}
        \elastic{-6.}{2.}
    \end{tikzpicture}
    \caption{The three main interactions of a proton with matter. An elastic Coulomb scattering (top) with
      the nucleus, a 
       non-elastic proton-nucleus collision (centre), and an
      {inelastic} Coulomb interaction with atomic
      electrons (bottom).
      \label{fig:atom}
    }
    \vspace{-30pt}
\end{wrapfigure}
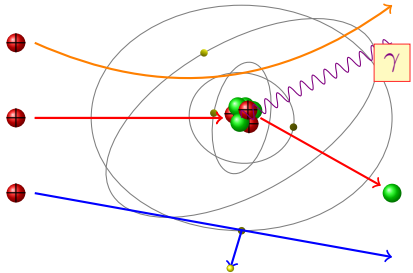
Consider now a  proton travelling through matter with an instantaneous  configuration $x= (\e, r,\omega)\in \mathcal{C}$. In order to describe the dynamics of an individual proton, we consider the notion of transport and three types of atomic interactions with surrounding matter. We refer the reader to \cite{NZ} for further details. 

\smallskip

\noindent{\bf Transport.} In the space between atoms, the proton moves in straight
lines in the direction $\omega$. 

  \smallskip
  
  \noindent{\bf Inelastic Coulomb interaction.} Protons {\color{black}lose} energy continuously as a consequence of a large number of collisions with orbital electrons. With instantaneous configuration $x = (\e, r,\omega)$, this occurs at rate (per unit track length) given by what is commonly referred to as the {\it cross section}\footnote{A consistent interpretation of the meaning of $\varsigma(x)$, and the other cross sections or rates discussed in this section, will also be specified in Section \ref{sde-model} in terms of stochastic modelling.} $\varsigma(x)$. In addition to this, one may also imagine an infinitesimally small scatter of the relatively massive proton from  each collision with the much smaller electrons. We will aggregate this small scatter with other types of small scatter (discussed below) and return to its treatment shortly. 

  \smallskip

\noindent{\bf  Elastic Coulomb scatter.} This event corresponds to the proton passing sufficiently close to an atomic nucleus to affect its trajectory. The term `elastic' refers to the fact that energy is conserved between the proton and the nucleus. Thus, this interaction results only in a change of direction of the incoming proton. We distinguish between two cases: big and small scatters. For this, we fix a cut-off, say $\delta > 0$, such that if $\omega$ denotes the incoming direction and $\omega'$ denotes the outgoing direction then the case $|\omega - \omega'| > \delta$ is termed a `big' (elastic Coulomb) scatter and the case $|\omega - \omega'| \le \delta$ is called a `small' scatter.
For big scatters, we use the cross section $\sigma_{\rm e}(x) \pi_{{\rm e}}(x; \dd \omega' )$ to mathematically describe these events. In this case, $\sigma_{\rm e}(x)$ denotes the rate at which a proton with configuration $x$ undergoes a `big' elastic Coulomb scatter, and the probability density $\pi_{{\rm e}}(x; \dd \omega' )$ denotes the probability that, given a proton with incoming direction $\omega$ undergoes a `big' elastic Coulomb scatter, the outgoing direction is $\omega'$.
In the (medical) physics literature, it is not uncommon to think of $\sigma_{\rm e}(x)$ as a rate per unit track length, which therefore captures the amount of atomic interactions it has in accordance with the medium in which it travels. 

\noindent {\bf Aggregated small scatter.}
We consider the net effect of all small scatter events to result in a Langevin-type component\footnote{We use the term Langevin here loosely to mean that velocity is driven by a Brownian motion, in this case a spherical Brownian motion on $\mathbb{S}_2$.} in the motion of the direction of motion $\omega$ of protons; cf.  Section 3.9.6 of \cite{Vassiliev}. The coefficient we will use for the Langevin-type behaviour of a proton  at instantaneous configuration $x$ is denoted $m(x)\geq0$. As with large scatter events, we assume there is no energy loss.

  \smallskip

\noindent{\bf Non-elastic proton-nucleus collision and scatter.}  A proton will encounter 
a collision with an atomic nucleus resulting in an  interaction as follows. Instantaneously, the incoming proton is absorbed into the nucleus, which then becomes excited, emits another proton in an outgoing direction $\omega'$ whilst simultaneously recoiling and de-energising by emitting a photon\footnote{It may also happen that neutrons and other heavy ions are produced as a result of this type of interaction.} ($\gamma$-ray). The energy of the outgoing proton relative to the energy of the incoming proton thus experiences a net loss. Roughly speaking, the lost energy is that which has been taken by the recoiling nucleus and emitted photon. 
 Non-elastic scatter occurs according to the cross section $\sigma_{{\rm ne}}(x)$, seen as a rate per unit track length, in  which case the configuration, $x$, of the incoming proton is transfered to an outgoing configuration $x' = (\e (1-u), r, \omega')$ 
with probability $\pi_{{\rm ne}}(x; \dd \omega', \dd u)$, where
 $u\in(0,1]$.

As an extreme case of a non-elastic proton-nucleus collision, it is possible that the incoming proton is absorbed into the nucleus without a new proton being emitted. In that case all of the energy of the incoming proton is absorbed and we may think of $\pi_{{\rm ne}}(\cdot; \cdot, \dd u)$ as having an atom at $u = 1$.

\smallskip

\begin{wrapfigure}{l}{6cm}
  \begin{center}
  \vspace{-10pt}

\includegraphics[width=5cm]{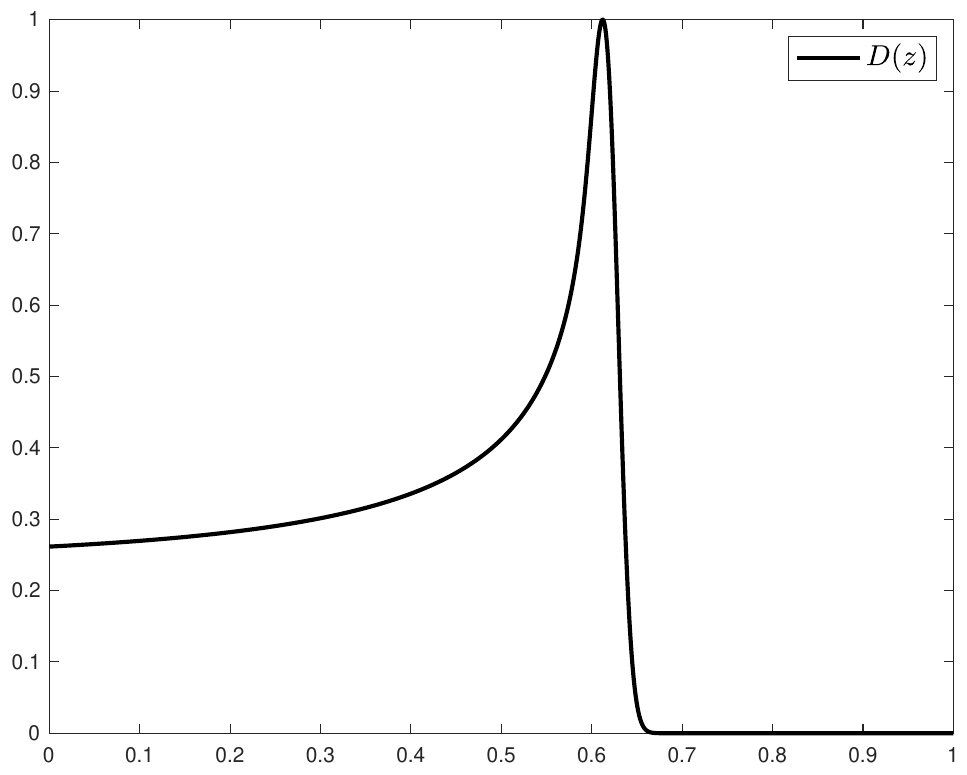}
\caption{An illustration of the Bragg Peak, $D(z)$, which plots energy deposition per unit cm against depth of proton beam into patient's tissue. In clinical settings, typical proton beams of energies range from 200 to 250 MeV are required to achieve the range of the order of 20-30 cm in tissue for treating deep-seated tumors.}
    \label{BP}
\end{center}
\vspace{-30pt}
\end{wrapfigure} 
 In the next section we will use the above 
description of the physical processes at play to describe what we call a {\it sequential proton track} and translate this into the language of a jump stochastic differential equation. By a sequential proton track, we mean the evolution in configuration space of an initial proton through to subsequent protons that continue its trajectory as a consequence of proton-nuclear interaction.

Ultimately, a sequential proton track will come to an end due to an absorption event occurring in an inelastic proton-nucleus collision. The terminus point of the sequential proton track is called the {\it range}. The less energy a sequential proton track has, the more interactions it will undergo. Qualitatively speaking, this means that the majority of the initial energy of the sequential proton track is deposited along its track close to the sequential proton track range.

\subsection{The anatomy of a proton beam and Bragg peak}

Whereas each sequential proton track is a random trajectory, we can think of a {\it proton beam} as being made up of many superimposed sequential proton tracks. Without loss of generality, we may assume that each sequential proton track behaves independently and therefore, up to a scaling factor, by the Law of Large Numbers,  the proton beam is nothing more than the average effect of an individual sequential proton track.

As alluded to above, the principle of  proton beam therapy  pertains to  the majority of the energy of each sequential proton track  being deposited in the patient's tissue around the range of each sequential proton track. 
The importance of this energy deposition in treating tumours is that the cellular structure of tissue (both healthy and unhealthy) is damaged by it. The nature of the damage comes about via 
two principal mechanisms. First, inelastic Coulomb interaction has an ionising effect on certain molecules within the cell, creating `unpleasant agents' within the cell which cause it to die. The second is that both high energy photons and protons that are consequent in the proton beam cause direct damage to the DNA of cells, preventing them from reproducing.


The Bragg peak is a  graphic which depicts the rate of energy deposition per unit length of proton beam against the length of the proton beam. Again, thinking of a a proton beam as the superposition of many sequential proton tracks, we may think of the Bragg Peak curve as representing the {\it average energy deposition per unit length of the proton beam.}

Figure \ref{BP} An illustration of the Bragg curve. The range of the beam is approximately where the beam peaks and this is where the maximal energy deposition occurs. Hence, in practice, the beam should be aligned so that the position of the Bragg peak aligns precisely with a targeted region within the cancerous tissue, while the rapid fall-off of dose near the end of range helps to spare normal tissues.

\subsection{This article}

To date we have been unable to find a systematic, holistic mathematical description for proton beam modelling which offers a full three-dimensional account of energy deposition.  In the well-cited \cite{bortfeld}, Bortfeld derives a parametric family of mathematical curves, which can be effectively fitted to empirical data to replicate experiments. Another key reference is \cite{NZ}, which describes a series of physical processes through a collection of mathematical descriptors that, together, explain the effect of proton beam energy deposition. Whilst these well-used articles do indeed use various mathematical ideas to describe components of the process that leads to energy deposition from proton beams,  there is no holistic and foundational mathematical understanding of how the Bragg curve emerges as a direct response to an integrated  mathematical description of subatomic random processes that are at play in the proton beam. In addition, there is no overarching mathematical model in the literature which is able to encompasses high-fidelity Monte Carlo simulations, such as those that can be performed by high-fidelity software such as Geant4.

Our objective in this article is therefore to establish a comprehensive mathematical modelling perspective that can be used in combination with existing and future modelling endeavours. The reader will see that we will need to appeal to some quite deep concepts from stochastic analysis in order to achieve our objective, and one may question whether such depth is needed. Indeed we understand that our approach is vulnerable to  the accusation that it is esoteric and not relevant to practice. In future work, we aim to demonstrate the comparative functionality of what we propose here with known numerical studies for phantoms. However, in this article, we are predominantly focused on the mathematical robustness of our completely new model. The assumptions involved in our model are extremely minimalistic relative to what reality demands, which will offer significant opportunity for model calibration in future work. Moreover, the nature of our model is conducive to understanding how to conduct and  interpret Monte Carlo simulations. 
\smallskip

\noindent{\it In particular, one of our main contributions will be to propose the concept of the Bragg curve in two and three dimensions and how it can be simulated using Monte Carlo, while also ensuring it is consistent with the one-dimensional object.} 
\smallskip

Another added value of this work is that it sets clear framework for developing different categories of inverse problems pertaining to  treatment planning via Bayesian optimisation, and to understanding the effect of treatment from secondary radiation. We discuss these two categories of problems in the conclusion of the paper.


\section{The mathematics of sequential proton tracks}\label{sde-model}
We need a core mathematical object that describes the physics that underlies sequential proton tracks, but at the same time is robust enough for us to develop analytical and numerical tools to work with in the context of proton beams. Due to the inherent randomness in the underlying physics of proton beams, the natural tool of choice is that of stochastic differential equations (SDEs). We stress that, in some sense, working with SDEs is inevitable given that their development has proved to be a pervasive approach to describe the evolution of a wide class of Markov processes and their connection to Feynman--Kac formulae; see for example \cite{Oksendal}. The Feynman--Kac representation is, of course, yet another mathematical theory that has materialised out of understanding the fundamental and yet random behaviour of subatomic particles and will be of paramount importance to us here as well.

\subsection{Track length Jump SDE}

It is usual to describe the evolution of an SDE with respect to time. However in the current setting, we follow the vast majority of physics and nuclear literature, and index our SDE by `track length'.  
To this end, we introduce  
 $Y =(Y_\ell, \ell \ge 0) = ( (\epsilon_\ell, r_\ell, \Omega_\ell), \ell\geq 0)$, where $\epsilon_\ell\in \mathcal{E}$, $r_\ell\in D$, and $\Omega_\ell\in\mathbb{S}_2$ for each $\ell\geq0$. The process  $Y$ represents the random evolution in configuration space, $\mathcal{C}$, of a sequential proton track. 
Recalling the cross sections mentioned in the introduction, we propose that a sequential proton track, $Y$, satisfies the jump stochastic differential equation

\begin{equation}
\boxed{
\begin{array}{rl}
\displaystyle\e_\ell& = \e_0 - \displaystyle\int_0^\ell \varsigma (Y_{{l}-})\dd l
- \displaystyle\int_0^\ell\int_{(0,1]}\int_{\mathbb{S}_2} u\e_{{l}-}{N}(Y_{{l}-}; \dd {l}, \dd \up', \dd u)\\
&\\
\displaystyle r_\ell &= r_0 +\displaystyle\int_0^\ell \Omega_{l-}\dd l  \\
&\\
\displaystyle\Omega_\ell& =
\Omega_0 -\displaystyle\int_0^\ell m(Y_l)^2\Omega_{l-}\dd l + \displaystyle\int_0^\ell m(Y_{l-})\Omega_{l-}\wedge\dd B_{l} \\
&\\
&\displaystyle\hspace{2cm}+ 
\int_0^\ell\int_{(0,1]}  
\int_{\mathbb{S}_2} (\Omega'-\Omega_{l-})
{N}(Y_{{l}-}; \dd {l}, \dd \up', \dd u)
\end{array}
\label{VSDE}
}
\end{equation}
where 
\begin{itemize}\itemsep0.8em
\item $\ell<\Lambda := \inf\{\ell>0: \epsilon_\ell = 0 \text{ or } r_\ell\not\in D\}$ is the track length after which the proton is absorbed or exits the domain, 
\item $Y_0 = (\epsilon_0, r_0, \Omega_0)\in  \mathcal{C}$ is the initial configuration of the proton,
\item $(B_\ell, \ell\geq0)$ is a standard Brownian motion on $\mathbb{R}^3$,
\item $m(x)
\geq0$ is the volatility in direction caused by small elastic Coulomb scattering,
\item $\varsigma(x)\geq0$ is the continuous rate of loss of energy due to inelastic Coulomb interaction and small elastic Coulomb scattering, and
\item for each $x\in \mathcal{C}, \ell\geq0, \up' \in \mathbb{S}_2, u\in(0,1]$, ${N}(x; \dd \ell, \dd \up', \dd u)$ is an optional random measure\footnote{The reader unfamiliar with the notion of optional random measures is referred to e.g.\ \cite{JS} or \cite{CE}.}  with  compensator $\sigma(x)\pi (x; \dd \up', \dd u)\dd l$ where,
\[
  \sigma(x) =\sigma_{{\rm e}}(x)+ \sigma_{{\rm ne}}(x)
\]
and
\begin{equation}
\pi (x; \dd \up', \dd u):=
\frac{\sigma_{{\rm e}}(x)}{\sigma(x) }\pi_{{\rm e}}(x; \dd \up' ) \delta_0(\dd u) +\frac{\sigma_{{\rm ne}}(x)}{\sigma(x) }\pi_{{\rm ne}}(x; \dd \up', \dd u).
\label{Npi}
\end{equation}
\end{itemize}

The first two integrals of the SDE for $\Omega_\ell$ describe a state-dependent Brownian motion on $\mathbb{S}_2$ (further described in Appendix A) with $m(x)
\geq0$ as its associated volatility. The interpretation of the  optional random measure ${N}$ in \eqref{VSDE} is that it captures discontinuities in the configuration space (specifically in energy and direction of motion) arising from both elastic and non-elastic Coulomb interactions. Specifically, when at $x=(\epsilon_{\ell-}, r_{\ell-}, \Omega_{\ell-}) \in \mathcal C$, interactions occur at rate $\sigma(x)$, at which point the sequential proton path jumps to configuration $(\epsilon_{\ell-}(1-u), r_{\ell-}, \Omega') $ with probability $\pi (x; \dd \up', \dd u)$. Moreover, proton absorption is embedded into the distribution $\pi$. Indeed, $N$ can be arranged so that  $q(x): =\sigma(x) \pi (x; \mathbb{S}_2, \{1\}) := \sigma(x) \int_{\mathbb S_2}\pi (x; \Omega', \{1\})\dd \Omega'$  is the killing rate at which there is proton absorption ending the sequential proton track.
Finally, as the proton beam is directional, we would expect the kernel $\pi$ to satisfy 
\begin{equation}
\int_{(0,1)}\int_{\mathbb{S}_2} (\Omega'-\Omega)\pi((\epsilon, r, \Omega); \dd\Omega', \dd u) = 0,
\label{meanshift}
\end{equation}
i.e.\ the direction of motion of a proton is unchanged on average.
 
 \smallskip

Below we state some of the key assumptions that we will place on \eqref{VSDE}. In order to state them, let us introduce the measure $\Sigma(\dd \Omega)$ for the surface measure on the unit sphere and ${\rm Leb}_S(\dd x) $ to denote Lebesgue measure on the  Borel set $S$ (which we will typically take to be $\mathcal{E}$, $D$, or $\mathbb{R}^3$ throughout this text). In addition, for any set $S$, we will write $B^+(S)$ to be the set of uniformly bounded measurable mappings from $S$ to $[0,\infty)$.

\begin{assumption}\label{ass1}\rm  
The following assumptions will be in force throughout this text.

\begin{itemize}
\item  The rate functions  $\sigma, \varsigma, m   \in B^+(\mathcal{C})$ are uniformly Lipschitz and uniformly bounded away from zero.

\item For each $x\in \mathcal{C}$, the distribution  $\pi(x; \dd \up', \dd u)$ is absolutely continuous with respect to  $\Sigma\times{\rm Leb}_{(0,1]}$,  with measurable density. With a slight abuse of notation, we denote the density by $\pi(x;  \up',   u)$.

\end{itemize}
\end{assumption}

Assumption \ref{ass1} is sufficient to ensure the existence and uniqueness of a   strong solution to \eqref{VSDE}; see Appendix B for further details. 
 Moreover, it is also standard that the solution to \eqref{VSDE} is a strong Markov process. As such, in the spirit of standard Markov process theory, we can introduce the probabilities 
$
\mathbb{P} = (\mathbb{P}_{x}, x = (\epsilon, r,\up)\in \mathcal{C})
$
to describe the distributions of the family of processes $Y=(Y_\ell, \ell\geq0)$ indexed by initial condition $x \in \mathcal{C}$.

In Section \ref{specifics} below we discuss some of the specific choices for the functionals $\sigma, \varsigma, m, \pi$ inspired by the physics literature that permit $Y$ to qualitatively replicate the behaviour of a sequential proton track, and how the solution to \eqref{VSDE} can be simulated.  For now, we press on to develop some of the underlying theory in relation to our generalisation of the Bragg peak in more than one dimension.

\section{The Bragg surface}

Whilst \eqref{VSDE} can be used to describe, path-by-path, the evolution in configuration space of individual sequential proton tracks, we need an associated mathematical quantity that will allow us to talk about the overall rate of energy deposition  of the proton beam. Given our earlier remarks identifying the proton beam as a multiple of the mean behaviour of a sequential proton track, as part of our  probabilistic point of view, we should similarly expect to identify energy deposition of the beam across space as an average.

To this end, let us introduce the following operator which may be seen as a path functional of the sequential proton track \eqref{VSDE}:
\begin{equation}
\mathtt{U}[f](x) = -\mathbb{E}_x\left[\int_0^\Lambda f(Y_{\ell-})\dd \epsilon_\ell\right], \qquad x\in  \mathcal{C},  f\in B^+( \mathcal{C}).
\label{Uopp}
\end{equation}
In words, for a sequential proton track issued with configuration $x\in  \mathcal{C}$, the quantity $\mathtt{U}[f](x) $ is the average deposition of energy in space, weighted by the test function $f\in B^+( \mathcal{C})$.

To help visualise the meaning of the definition \eqref{Uopp}, suppose $A\subset D$  is a  small voxel  of interest in  space. If we define $f$ to be  $f(\e, r, \omega) = \mathbf{1}_{\mathcal{E}\times A\times \mathbb{S}_2}(\epsilon, r,\omega)$, then the random variable 
\[
-\int_0^\Lambda f(Y_{\ell-} )\dd \epsilon_\ell
\]
is the total energy deposited in $A$ by the sequential proton track $Y$. 
If, moreover,  $n_0$ is the number of independent sequential proton tracks that constitute an individual proton beam that are all initiated with initial configuration $x\in\mathcal{C}$, then $n_0\mathtt{U}[\mathbf{1}_{\mathcal{E}\times A\times \mathbb{S}_2}](x) $ gives precisely the average energy deposition of the beam in $A$.

 Note that the negative sign in \eqref{Uopp} is necessary as  the energy of a sequential proton track is decreasing, while energy deposition is a non-negative quantity. Decrements in the energy of a sequential proton track are thus converted to an increment in energy deposition via a change of sign. We also note that $\mathtt{U}[f](x)$ is well defined as a finite object when we take account of the simple inequality
\[
\mathtt{U}[f](x) \leq  
\epsilon_0 \sup_{y\in \mathcal{C}}f(y) \le \mathtt{e}_\mathtt{max}\sup_{y\in \mathcal{C}}f(y) <\infty,
\]
where $\epsilon_0$ is the initial energy corresponding to the initial configuration  $x = (\e_0, r_0, \omega_0)$.

\subsection{Absolute continuity of energy deposition}\label{AC}
Recall that energy depletes continuously at rate $\varsigma$ and discontinuously by a proportion $u\in(0,1]$ at rate $\sigma (x) \pi (x; \mathbb{S}_2, \dd u)$, which includes  the possibility of the sequential proton track reaching its range at rate $q(x)=\sigma (x) \pi (x; \mathbb{S}_2, \{1\})$.  The compensation formula for optional random measures thus yields, for $f \in B^+( \mathcal{C})$, $x\in \mathcal C$, 
\begin{align}
\mathtt{U}[f](x) &=  \mathbb{E}_x\left[\int_0^{\kappa_D} f(Y_{\ell}) {\rm e}^{-\int_0^\ell q(Y_l)\dd l }\varsigma(Y_{\ell})\dd \ell\right] \\
&\hspace{1cm}+ \mathbb{E}_x\left[\int_0^{\kappa_D} \int_{(0,1)}f(Y_\ell)u\epsilon_{\ell}{\rm e}^{-\int_0^\ell q(Y_l)\dd l }\,  \sigma (Y_{\ell}) \pi (Y_{\ell}; \mathbb{S}_2, \dd u) \dd \ell\right] \notag\\
&\hspace{2cm} +  \mathbb{E}_x\left[\int_0^{\kappa_D}f(Y_\ell) \epsilon_{\ell}q(Y_\ell){\rm e}^{-\int_0^\ell q(Y_l)\dd l }\,  \dd \ell\right]\notag\\
&=  \mathbb{E}_x\left[\int_0^{\Lambda} f(Y_{\ell}) \varsigma(Y_{\ell})\dd \ell\right] + \mathbb{E}_x\left[\int_0^{\Lambda} \int_{(0,1]}f(Y_\ell)u\epsilon_{\ell}\,  \sigma (Y_{\ell}) \pi (Y_{\ell}; \mathbb{S}_2, \dd u) \dd \ell\right] \notag\\
&= \int_{ \mathcal{C}} f(\epsilon, r, \Omega)\left\{  \varsigma(\epsilon, r, \Omega) 
+\epsilon\int_{(0,1]}u \sigma (\epsilon, r, \Omega) \pi ((\epsilon, r, \Omega); \mathbb{S}_2, \dd u)
 \right\}
\mathtt{R}(x; \dd \epsilon\, \dd r\, \dd\Omega),
\label{killinginout}
\end{align}
where $\kappa_D=\inf\{\ell>0: r_\ell \not\in D\}$, and 
\[
\mathtt{R}(x; \dd y) = \mathbb{E}_x\left[\int_0^{\Lambda} \mathbf{1}_{(Y_\ell\in \dd y)} \dd \ell\right], \qquad x, y \in \mathcal{C},
\]
is the resolvent of $Y$ with respect to track length. The following result, despite  seeming to be a  technical distraction, is crucial to how we will describe the rate of energy deposition, thereby generalising the notion of Bragg peak. In what follows, we use $C^\infty$ to denote the set of smooth functions on $\mathcal C$ and $C^\infty_b \subset C^\infty$ to denote those functions $f$ in $C^\infty$ such that for any multi-index $\alpha$, $\partial_x^\alpha f(x)$ is bounded.

\begin{theorem}\label{density}
Suppose that  Assumption \ref{ass1} holds and that $\varsigma, m \in C_b^\infty$. Further assume that $\varsigma(\epsilon, r, \varphi, \vartheta)^{1/2} + \partial^2_{\vartheta\vartheta}(\varsigma(\epsilon, r, \varphi, \vartheta)^{1/2})$ is bounded away from $0$ and $\infty$, where we have written $(\varphi, \vartheta)$ in place of $\omega$ using polar coordinates.
Then \eqref{VSDE} has a strong solution and the resolvent $\mathtt{R}(x; \dd y)$ has a smooth density with respect to ${\rm Leb}_\mathcal{E}\times{\rm Leb}_D\times \Sigma $. That is, for each $x\in  \mathcal{C}, \epsilon\in\mathcal{E}, r\in D, \Omega\in\mathbb{S}_2$, we have
\begin{equation}
\mathtt{R}(x; \dd \epsilon \, \dd r\, \dd \Omega )= \mathtt{r}(x,(\epsilon, r,\Omega))\, {\rm Leb}_\mathcal{E}( \dd\epsilon)\, {\rm Leb}_D(\dd r)\,\Sigma(\dd\Omega),
\label{BBoccupation}
\end{equation}
where $\mathtt{r}(x, (\epsilon, r,\Omega))$ is a density in $\mathcal{C}^\infty$.
\end{theorem}

\begin{remark}\label{rem:smoothness}\rm
It is possible to relax the regularity assumptions on the cross sections in the above theorem at the expense of a less regular density. For example, one could instead assume that the cross sections are only $k$ times continuously differentiable for some $k \ge 2$, which would result in the existence of a density $\mathtt r(x, \cdot)$ that is $\ell$ times continuously differentiable for some $\ell \le k$. This is an unnecessary luxury as the current assumptions comfortably allow for the requirements of any real-world model.
It therefore lies beyond the purpose of this paper and we leave the details to the interested reader. 
\end{remark}

\begin{remark}\rm We also note that, if we define the   resolvent operator 
\[
\mathtt{R}[f](x) = \int_{\mathcal{C}} f(\epsilon, r,\Omega)\mathtt{r}(x,(\epsilon, r,\Omega))\, {\rm Leb}_\mathcal{E}( \dd\epsilon)\, {\rm Leb}_D(\dd r)\,\Sigma(\dd\Omega), \qquad  x\in\mathcal{C},
\] 
then a  standard Feyman--Kac heuristic dictates that, in an appropriate analytical sense (which we do not discuss further), $\mathtt{R}[f](x) $ is associated to the equation
\[
\mathcal{L}\mathtt{R}[f](x)+ f(x) = 0, \qquad x\in \mathcal{C},
\]
with $\mathtt{R}[f](\epsilon, r, \Omega) = 0$ whenever $r\in \partial D$, and where, for $x= (\epsilon, r, \Omega)\in\mathcal{C}$,
\begin{align*}
\mathcal{L}g(x) &= \Omega\cdot\nabla_r g(x) - \varsigma(x)\frac{\partial}{\partial \epsilon}g(x) + m(x)\triangle_\Omega g(x)\notag\\
&\hspace{1cm}
+ \sigma(x)\int_\mathcal{C} \Big(g((1-u)\epsilon, r, \Omega')-g(x)\Big)\pi (x; \dd \up', \dd u),
\end{align*}
 $\nabla_r$ is the gradient operator with respect to the variable $r\in D$, and $\triangle_\Omega$ is the Laplace--Beltrami operator with respect to the variable $\Omega\in\mathbb{S}_2$. Similar observations for both general and other specific forms of Boltzmann transport equations have been recently been discussed in \cite{otherSDE} and \cite{nukebook}.
\end{remark}

Thanks to the conclusion of Theorem \ref{density},  under its assumptions, we can now associate a density to the operator $\mathtt{U}: B^{+}(\mathcal{C})\to B^{+}(\mathcal{C})$. In particular we can identify the density $\mathtt{u}(x, y)$, $x,y\in \mathcal{C}$ via
\begin{equation}
\mathtt{U}[f](x) = \int_{\mathcal{C}} f(\epsilon, r,\Omega)\mathtt{u}(x,(\epsilon, r,\Omega))\, {\rm Leb}_\mathcal{E}( \dd\epsilon)\, {\rm Leb}_D(\dd r)\,\Sigma(\dd\Omega),
\label{putinBB}
\end{equation}
where 
\[
\mathtt{u}(x, (\epsilon, r, \Omega))=
\left\{ \varsigma(\epsilon, r, \Omega) 
+\e \sigma (\epsilon, r, \Omega)\int_{(0,1]} u \pi ((\epsilon, r, \Omega); \mathbb{S}_2, \dd u)\right\}
\mathtt{r}(x,(\epsilon, r, \Omega)).
\]

\begin{lemma}\label{bdddensity}
The density $\mathtt{u}(x,y)$, $x,y\in\mathcal{C}$ is uniformly bounded.
\end{lemma}
\begin{proof}
Since $\varsigma$ is assumed uniformly bounded away from zero, we have from the form of \eqref{VSDE} that 
$
\Lambda \leq {\e_0}/{\underline{\varsigma}}
$
where $\underline\varsigma = \inf_{x\in\mathcal{C}}\varsigma(x)>0$. As such, $\mathtt{r}(x, y)$ is uniformly bounded. The proof now follows from Assumption \ref{ass1}.
\end{proof}
\subsection{Connection to the Bethe--Bloch formula}

In the theory of subatomic  nuclear physics, the Bethe--Bloch formula describes the mean energy loss per distance travelled of  charged particles (typically protons, alpha particles, atomic ions) traversing matter. Alternatively this can be seen as the stopping power of the material through which they travel.

The Bethe--Bloch formula, in its most general form, tells us that for a particle with energy $\epsilon$, travelling a distance $\ell$ into a target material,
\begin{equation}
-\left\langle {\frac {\dd\e}{\dd\ell}}\right\rangle=
4\pi
\rho
N_A
r_e^2 
m_e
c^2 
\frac{Z}{A}\frac{z^2}{\beta^2}
\left[\ln \left({\frac {2m_{e}c^{2}\beta ^{2}}{I (1-\beta ^{2})}}\right)-\beta^{2} -\frac{\delta}{2}-\frac{C}{Z}\right],
\label{realBB}
\end{equation}
where, apart from the minus sign, we read the left-hand side as a single notation meaning the mean energy loss per distance travelled of a charged particle. For the terms on the right-hand side, $\rho$ is the density of the material, $N_A$ is Avogadro’s number, $r_e$ is the classical electron radius, $m_e$ is the mass of an electron, $c$ is speed of light, $Z$ is the atomic number of the absorbing medium, $A$ is the atomic weight of the absorbing material, $z$ is the charge of the projectile, $\beta = \upsilon/c$ where $\upsilon$ is the speed of the projectile, $I$ is the mean excitation energy, $\delta$ is the density corrections arising from the shielding of remote electrons by close electrons and will result in a reduction of energy loss at higher energies, and $C$ is the shell correction item, which is important only for low energies where the particle velocity is near the velocity of the atomic electrons. Some of the  terms in the Bethe--Bloch equation involve relativistic theory and quantum mechanics and need to be considered when very high or very low proton energies are used in calculations.  
 
Although extremely involved, the take-home message from the Bethe--Bloch formula is that the projectile’s characteristics govern its energy loss rate. That is to say, 
each of the quantities $\rho, I, N_A, r_e, m_e, z, Z, A, c, \beta, \upsilon, \delta, C$ are either physical constants, or functions of  the particle's  energy-position-direction configuration, $(\epsilon, r, \Omega)$. The connection between this formula and our SDE \eqref{VSDE} is given through the simple relation
\begin{equation}
\left\{ \varsigma(\epsilon, r, \Omega) 
+\e \sigma (\epsilon, r, \Omega)\int_{(0,1]} u  \pi ((\epsilon, r, \Omega); \mathbb{S}_2, \dd u)\right\} = - \left\langle {\frac {\dd \epsilon} {\dd\ell}} \right\rangle,
\label{BB1}
\end{equation}
for $(\epsilon, r,\omega)\in \mathcal{C}$. Indeed, both the left- and right-hand sides correspond to the average rate of energy loss per unit length, per particle.
As such, the expression \eqref{putinBB} now reads 
\[
\mathtt{U}[f](x) =- \int_{\mathcal{C}} f(\epsilon, r,\Omega) \left\langle {\frac {\dd \epsilon} {\dd\ell}} \right\rangle\mathtt{r}(x, (\epsilon, r,\Omega))\, {\rm Leb}_\mathcal{E}( \dd\epsilon)\, {\rm Leb}_D(\dd r)\,\Sigma(\dd\Omega).
\]
That is, the total energy deposition, weighted by the test function $f$, is computed by integrating $f$ against a density consisting of the instantaneous Bethe--Bloch rate of energy deposition per unit track length multiplied by the mean occupation density of tracks over all possible configurations.



\subsection{The Bragg manifold and Bragg surface}\label{BM}
We are interested in the rate at which energy is deposited at individual spatial sites. Let us assume henceforth that the assumptions of Theorem \ref{density} are in force, permitting the existence of a density $\mathtt{u}$.
 There are at least two   different ways in which we can formulate this. 
For example, for a {\it fixed}  direction $\omega\in\mathbb{S}_2$, we might be interested in  the average rate at which energy is deposited at individual spatial sites in the direction of $\omega$. 
 This corresponds to setting (with an abuse of notation) $f(\epsilon, r,\omega) = f(r)\in B^+( D)$ and calculating, with the help of the dominated convergence theorem, 
 for $x\in\mathcal{C}$, 
\begin{align}
D_\omega[{f}](x)&: = - \lim_{\delta\to0}\frac{1}{\delta} 
\mathbb{E}_x\left[\int_0^\Lambda \Big({f}(r_\ell + \delta\omega)-{f}(r_\ell)\Big) \dd \epsilon_\ell\right]\notag\\
&=- \mathbb{E}_x\left[\int_0^\zeta \omega\cdot \nabla_r {f}(r_\ell) \dd \epsilon_\ell\right] \notag\\
&= \int_{ \mathcal{C}}\omega\cdot\nabla_r{f}(r)\,\mathtt{u}(x, (\e, r,\Omega))\, {\rm Leb}_\mathcal{E}( \dd\epsilon)\, {\rm Leb}_D(\dd r)\,\Sigma(\dd\Omega),
\label{eg1}
\end{align}
where we assume that $f$ is sufficiently smooth to pass through the gradient operator (e.g.\ $f$ is continuously differentiable  in $D$ with bounded derivatives). 
Another case in point is that we may also be interested in the average rate at which energy is deposited in the instantaneous direction of travel of the sequential proton track. In that case, we may replace the fixed $\omega$ by the variable $\Omega_{\ell}$ and a similar calculation yields, for $x\in\mathcal{C},$ 
\begin{align}
D[{f}](x)&:=
- \lim_{\delta\to0}\frac{1}{\delta} 
\mathbb{E}_x\left[\int_0^\Lambda \Big({f}(r_\ell + \delta \Omega_\ell)-{f}(r_\ell)\Big) \dd \epsilon_\ell\right]\notag\\
& = \int_{\mathcal{C}} \Omega\cdot\nabla_r{f}(r )\,\mathtt{u}(x, (\e, r,\Omega))\, {\rm Leb}_\mathcal{E}( \dd\epsilon)\, {\rm Leb}_D(\dd r)\,\Sigma(\dd\Omega),
\label{eg2}
\end{align}
where, again, we assume that $f$ is sufficiently smooth. 

To put this in a common framework, let us introduce some notation. For convenience, let us write  for short 
\[
\langle f , \mathtt{u}(x,\cdot)\rangle :=  \int_{\mathcal{C}} f(\e, r, \Omega)\,\mathtt{u}(x, (\e, r,\Omega))\, {\rm Leb}_\mathcal{E}( \dd\epsilon)\, {\rm Leb}_D(\dd r)\,\Sigma(\dd\Omega),
\]
where $\langle\cdot,\cdot\rangle$ stipulates an $L^2({\rm Leb}_\mathcal{E}\times{\rm Leb}_D\times \Sigma)$ functional inner product in the usual way on $\mathcal{C}$ with respect to the reference measure ${\rm Leb}_\mathcal{E}\times{\rm Leb}_D\times \Sigma$.

Thanks to the existence of the density $\mathtt{u}$, the two examples \eqref{eg1} and \eqref{eg2} can now be more conveniently written as 
\begin{equation}
D_\omega[{f}](x) =-\langle \mathtt{T}^{\omega} [ {f}], \mathtt{u}(x,\cdot) \rangle\quad \text{ and }\quad D[{f}](x) =
-\langle \mathtt{T}[ {f}], \mathtt{u}(x,\cdot) \rangle,
\label{2egs}
\end{equation}
where we have defined the $L^2({\rm Leb}_\mathcal{E}\times{\rm Leb}_D\times \Sigma)$ transport  operators $\mathtt{T}^\omega[ f](\epsilon, r,\Omega) = -\omega\cdot\nabla_r f(\epsilon, r, \Omega)$ (transport in the {\bf exogenous} direction $\omega$) and $\mathtt{T} [ f](\epsilon, r,\Omega) =- \Omega\cdot\nabla_r f(\epsilon, r, \Omega)$ (transport in the  {\bf endogenous} direction $\Omega$), for appropriate test functions $f$.

To be more precise about what we mean by `appropriate test functions', let us introduce some further notation.
We define $\Gamma^+$ and $\Gamma^-$ via
\[
\Gamma^\pm = \{ (r, \Omega) \in \partial D \times \mathbb S_2 : \pm \Omega\cdot\mathbf{n}_r>0\},
\]
where,  recalling that $D$ is convex, $\mathbf{n}_r$ is the outward facing normal to $r\in\partial D$. 
From this we may define  ${\rm Dom}^+$ and ${\rm Dom}^-$ via
\begin{align*}
{\rm Dom}^\pm : = \{f\in L^2({\rm Leb}_\mathcal{E}\times{\rm Leb}_D\times \Sigma)\,: \,&\Omega\cdot\nabla_r f(\e,r,\Omega)\in L^2({\rm Leb}_\mathcal{E}\times{\rm Leb}_D\times \Sigma)\notag\\
&\hspace{4cm} \text{ and } f|_{\mathcal{E}\times\Gamma^\pm} =0\}.
\end{align*}

\begin{lemma}\label{switchlemma}
For $g\in {\rm Dom}^+$ and $(\e, r, \Omega) \in \mathcal C$, define the 
operators 
\[
\hat{\mathtt{T}}^\omega[g](\epsilon, r,\Omega) = \omega\cdot\nabla_r g(\epsilon, r,\Omega) \quad \text{and} \quad \hat{\mathtt{T}}[g](\epsilon, r,\Omega) = \Omega\cdot\nabla_r g(\epsilon, r,\Omega).
\] 
Then, for $f \in {\rm Dom}^-$, we may re-write the derivatives in \eqref{2egs} as 
\begin{equation}
 D_\omega[{f}](x) = -\langle {f}, \hat{\mathtt T}^\omega[\mathtt{u}(x,\cdot)] \rangle\quad \text{ and }\quad
D[{f}](x) =-\langle  {f}, \hat{\mathtt{T}}[\mathtt{u}(x,\cdot) ]\rangle,\qquad x\in \mathcal{C},
\label{2egs*}
\end{equation}
respectively.
 \end{lemma}

\begin{proof}Let us consider the case of $D[f](x)$. The case of $D_\omega[f](x)$ can be proved similarly.
We first show that $\mathtt{T}$ and $\hat{\mathtt{T}}$ are adjoint when restricted to appropriate classes of functions. To this end, for $g \in {\rm Dom}^+$ and $f\in \{h \in L^2({\rm Leb}_\mathcal{E}\times{\rm Leb}_D\times \Sigma) : \Omega \cdot \nabla_r h(\e, r, \Omega) \in L^2({\rm Leb}_\mathcal{E}\times{\rm Leb}_D\times \Sigma)\}$,
\begin{align} 
\langle  \hat{\mathtt{T}}[g] , f \rangle &= \int_{\mathcal{E}\times\partial D\times \mathbb{S}_2} (\Omega\cdot {\bf n}) g(\e, r,\Omega)f(\e, r,\Omega)  \,{\rm d}S\,{\rm Leb}_\mathcal{E}( \dd\epsilon)\, \Sigma(\dd\Omega)  -\langle  g, \Omega\cdot\nabla_r  f \rangle\notag\\
&=\int_{\mathcal{E}\times\partial D^-\times \mathbb{S}_2} (\Omega\cdot {\bf n}) g(\e, r,\Omega)f(\e, r,\Omega)  \,{\rm d}S\,{\rm Leb}_\mathcal{E}( \dd\epsilon)\, \Sigma(\dd\Omega)  -\langle g,  \Omega\cdot\nabla_r  f\rangle,
\label{nabladual}
\end{align}
where $\dd S$ is the surface measure on $\partial D$ and ${\bf n}$ is the associated outward normal. If we further assume that $f|_{\mathcal{E}\times\Gamma^-} =0$ so that $f\in{\rm Dom}^-$, then we see from \eqref{nabladual} that 
\[
\langle  \hat{\mathtt{T}}[g], f \rangle = \langle g, \mathtt{T}[f] \rangle.
\]

Next, if $(\e,r,\Omega)\in\mathcal{E}\times\Gamma^+$, then from the path regularity\footnote{Here we mean that a solution of \eqref{VSDE} which is initiated with  $r\in\partial D$ and $\Omega\cdot\mathbf{n}>0$ (i.e.\ direction of motion is pointing outwards), then the $r_\ell$ component of the SDE will immediately visit the exterior of $D$ and hence will be killed.} of solutions to \eqref{VSDE} it follows that $\mathbb{P}_{(\e,r,\Omega)}(\Lambda = 0)=1$ and hence $\mathtt{r}(\e,r,\Omega)=\mathtt{u}(\e,r,\Omega)=0$. Thus, together with the conclusions of Lemma \ref{bdddensity} and Theorem \ref{density}, it follows that $\mathtt{u}\in {\rm Dom}^+$. Hence, for $f \in {\rm Dom}^-$, we have 
\[
D[{f}](x) =-\langle \mathtt{T}[ {f}], \mathtt{u}(x,\cdot) \rangle 
=-\langle  {f}, \hat{\mathtt{T}}[\mathtt{u}(x,\cdot) ]\rangle,
\]
as required.
\end{proof}

The above discussion motivates the following definition, which generalises the concept of the Bragg curve to the setting of what we will call a {\it Bragg manifold}; a functional of the full physical configuration space which describes the rate of energy deposition given the initial beam configuration.
\begin{definition}[Bragg manifold]
For $x \in \mathcal C$, we define the associated \emph{Bragg manifold} via
\[
\boxed{
\mathtt{b}(x,(\epsilon, r,\Omega)) = - \Omega\cdot\nabla_r\mathtt{u}(x, (\epsilon, r,\Omega)), \quad (\e, r, \Omega) \in \mathcal C.
}
\]
\end{definition}

As alluded to above, the Bragg manifold is the average rate of directional energy deposition at configuration $(\epsilon, r,\Omega)$ in the sequential proton track, when issued from a configuration $x\in \mathcal{C}$.


 Suppose that a beam is fired into the domain $D$ in the direction $\omega_0$ (pointing to the interior of $D$) at the point $r_0\in \partial D$ and with initial energy $\epsilon_0$. To take account of the fact that the proton beam is made up of the superposition of many sequential proton tracks, and that there will be some variability around the entry data $(\epsilon_0, r_0,\omega_0)$, suppose that $\eta(\dd x)$ is a measure focused on a small domain $\mathcal{E}_0\times \mathcal{D}_0\times \mathcal{S}_0$, where $\mathcal{E}_0$ is centred around the desired `per track' energy $\epsilon_0$, $\mathcal{D}_0\subset\partial D$ is a small domain of track entry points surrounding $r_0$ and $\mathcal{S}_0\subset \mathbb{S}_2$ is a small neighbourhood of $\omega_0$ in $\mathbb{S}_2$ accounting for variation in the sequential proton track directions. The measure $\eta$ accounts for the average density of sequential proton tracks initiated at the boundary point centred at $r_0$  and is called the {\it proton beam input measure}. Accordingly this leads us to the definition of a Bragg surface.

\begin{definition}[Bragg surface]\label{def:BS}
For a given proton beam input density $\eta$, we define the associated {\em Bragg surface} via
\begin{equation}
\boxed{
\mathtt{B}(r):
= \langle\eta, \mathtt{b}(\cdot , \mathcal{E}, r, \mathbb{S}_2)\rangle, \quad r \in D,
 }
 \label{BS}
\end{equation}
where we are writing $g(x;\mathcal{E}, r, \mathbb{S}_2)$ as shorthand for $\int_{(\epsilon, \Omega)\in\mathcal{E}\times \mathbb{S}_2}g(x;\epsilon, r, \Omega)\,{\rm Leb}_\mathcal{E}( \dd\epsilon)\, \Sigma(\dd\Omega) $.
\end{definition}

We conclude this section by remarking that our definition for both the Bragg manifold and Bragg surface are new. Moreover, we have been unable to locate any comparable concept within existing literature, other than what is known for one-dimensional models, i.e.\ Bortfeld's approximation \cite{bortfeld}.

\section{Consistency with the Bortfeld derivation in one dimension}

Let us show the consistency of our definition of Bragg surface with what is currently understood as the Bragg curve and in particular, how this aligns with the analytical approximation of the Bragg curve given in \cite{bortfeld}.

\subsection{One-dimensional jump SDE} 

In order to make the comparison,   we  must first reconsider our calculations using \eqref{VSDE} in one dimension, which requires some minor adjustment.

In this setting the domain $D$ is an interval, say  $[0,a]$, and the deposition of energy occurs along it according to some point process. Accordingly, we may assume that particles travel only in one direction. That is to say, we no longer have a need for the variable $\Omega$ (which in one dimension would have to be valued in the two-point set $\{-1,1\}$). Indeed, scattering at elastic and non-elastic events would otherwise consist of  changing direction of movement along the interval $[0,a]$. Irrespective of whether or not this feature is included, the positions at which energy is deposited during non-elastic interactions is still described by a point process, now on $[0, a]$. Similarly, the Brownian component may be ignored as its inclusion only changes the positions  on $[0,a]$ at which energy is deposited, and hence its effect can be accounted for by the rate function of the energy deposition point process.

Because of the reduction in the number of  variables to  depth, $r$, and energy, $\epsilon$,  the SDE \eqref{VSDE} becomes significantly simpler and can now be written in the form
\begin{equation}
\boxed{
\epsilon_\ell = \epsilon_0 -\int_0^\ell  \varsigma(\epsilon_{{r}-}, {r}) \dd {r} - \int_0^\ell \int_{(0,1]}u\epsilon_{{r}-}{N}((\epsilon_{{r}-}, r); \dd {r},  \dd u),
\label{SDE3}
}
\end{equation}
for $0\leq \ell<\Lambda $, 
where 
\[
\Lambda = \inf\{\ell>0: \e_\ell = 0\}\wedge a
\]
 is the total track length, $\varsigma(\e, r)$ is the continuous rate of energy loss due to inelastic Coulomb interaction,
${N}((\epsilon_{\ell-}, r) ; \dd {r},  \dd u)$ is an optional random measure with previsible compensator 
$\sigma(\epsilon_{{r}-}, r)\pi((\epsilon_{{r}-}, r); \dd u)\dd r$
describing the deposition of a proportion  $u\in(0,1]$ of the existing energy in a small $\dd {r}$ with total rate $ \sigma_{{\rm ne}}(\epsilon_{{r}-}, r) $ and probability distribution $\pi( (\epsilon_{{r}-}, r); \dd u)$. As in the higher-dimensional model, we include the possibility that the sequential proton track reaches its range through a killing event at rate  
$\sigma(\epsilon_{{r}-}, r)\pi((\epsilon_{{r}-}, r); \{1\})\dd r$.
An important difference with the higher-dimensional model \eqref{VSDE} is that  the variable corresponding to track length is now, itself, a spatial variable.

Now taking $f: \mathcal E \times [0, a] \to\mathbb[0,\infty)$ and following similar calculations to \eqref{killinginout}, the total expected energy deposition along the sequential proton track is given by
\begin{align}
\mathtt{U}[f](\epsilon_0, 0) 
&= -\mathbb{E}_{(\epsilon_0, 0)}\left[\int_0^{ \Lambda} f(\epsilon_{{r}-}, r)\dd \epsilon_{r}\right] \notag\\
&= \mathbb{E}_{(\epsilon_0, 0)}\left[\int_0^{ \Lambda} f(\epsilon_{{r}}, r)\left\{ \varsigma(\epsilon_{{r}}, r)+ \sigma (\epsilon_{{r}}, r)  \int_{(0,1]} u \epsilon_{{r}}\pi ((\epsilon_{{r}}, r); \dd u)
 \right\}\dd {r}\right] \notag\\
&=-\int_0^a\int_0^{\epsilon_0} f(\epsilon, r)\left\langle {\frac {\dd \epsilon} {\dd r}} \right\rangle\mathbb{P}_{(\epsilon_0, 0)}(\epsilon_{r}\in \dd \e, \, {r}< \Lambda)\dd {r},
\label{1DEdep}
\end{align}
where we have also used \eqref{BB1}.
As alluded to above, because the variable $r_\ell$ has been replaced by depth ${r}$, which is now both a spatial variable and track length,  we lose the link to the resolvent density.

If we take $f(\epsilon, r) = \mathbf{1}_{(r\leq \ell)}$ in \eqref{1DEdep} for some depth $\ell\in[0,a]$, we see that the total average energy deposited up to depth $\ell$ along the sequential proton track is given by 
\begin{align}
&\mathtt{U}[\mathbf{1}_{(\cdot\leq \ell)}](\epsilon_0, 0) =-\int_0^\ell\int_0^{\epsilon_0}  
\left\langle {\frac {\dd \epsilon} {\dd r}} \right\rangle
\mathbb{P}_{(\epsilon_0, 0)}(\epsilon_{r}\in \dd\e, \, {r}< \Lambda)\dd {r}.
\label{comparebortfeld}
\end{align}
On the other hand, the left-hand side of \eqref{1DEdep} directly yields, for $\ell<a$,
\begin{align}
&\mathtt{U}[\mathbf{1}_{(\cdot\leq \ell)}](\epsilon_0, 0) =-\mathbb{E}_{(\epsilon_0, 0)}\left[\int_0^{ \Lambda} \mathbf{1}_{(r\leq \ell)}\dd \epsilon_{r}\right] =\mathbb{E}_{(\epsilon_0, 0)}[\epsilon_0-\e_{\ell\wedge\Lambda}] = 
\epsilon_0-\mathbb{E}_{(\epsilon_0, 0)}[\e_{\ell}; \ell < \Lambda] .
\label{comparebortfeld1}
\end{align}
\noindent Equating \eqref{comparebortfeld} and \eqref{comparebortfeld1} and differentiating with respect to $\ell$, it follows that the rate of energy deposition per unit length is given by
\begin{align*}
-\int_0^{\epsilon_0}\left\langle {\frac {\dd \epsilon} {\dd\ell}} \right\rangle\mathbb{P}_{(\epsilon_0, 0)}(\epsilon_{\ell}\in \dd\e, \, \ell< \Lambda)= -\frac{\dd }{\dd \ell} \mathbb{E}_{(\epsilon_0, 0)}[\e_{\ell}; \ell < \Lambda].
\end{align*}

If $\eta$ is the average number of sequential proton tracks that are released at the boundary, then the total average rate of change of energy deposition with respect to depth, i.e. the Bragg curve, is given by 
\begin{equation}
\boxed{
\mathtt{B}(\ell):= -\eta\int_0^{\epsilon_0}\left\langle {\frac {\dd \epsilon} {\dd\ell}} \right\rangle\mathbb{P}_{(\epsilon_0, 0)}(\epsilon_{\ell}\in \dd\e, \, \ell< \Lambda) = -\eta \frac{\dd }{\dd \ell} \mathbb{E}_{(\epsilon_0, 0)}[\e_{\ell}; \ell < \Lambda].
}
\label{BP1d}
\end{equation}

If we suppose  that $\mathbb{P}_{(\epsilon_0, 0)}(\epsilon_{r}\in \dd\e, \, {r}< \Lambda) =: \mathtt{g}(\epsilon, {r}) \dd \epsilon$, then additionally note that  we can write
\begin{equation}
\boxed{
\mathtt{B}(\ell):=-\eta\int_0^{\epsilon_0}\left\langle {\frac {\dd \epsilon} {\dd\ell}} \right\rangle\mathtt{g}(\epsilon, \ell) \dd \e.
}
\label{BP1d2}
\end{equation}
If we compare this expression with the Bragg surface in \eqref{BP1d2}, the integral with respect to $\eta$ in \eqref{BS} is replaced by the simple factor $\eta$ in \eqref{BP1d2}. Thereafter a direct comparison with the analytical structure of \eqref{BS} is difficult to make directly because in one dimension, the track-length parameter becomes the same as the one spatial dimension. For example,  whilst $\mathtt{g}$ plays a similar role to $\mathtt{r}$, the operator $\Omega\cdot\nabla_r$ no longer appears.


\subsection{Alignment with the Bortfeld model}
Bortfeld \cite{bortfeld} begins his description of the Bragg curve by claiming that the total 
energy fluence $\Psi(\ell)$ at a depth $\ell$ can be written in the form 
\[
\Psi(\ell) = \Phi(\ell){E}(\ell),
\]
where $\Phi(\ell)$ is particle fluence and $E(\ell)$ is remaining energy at depth $\ell$ (per particle).  
Up to a multiplicative constant (which converts quantities to energy release per unit mass), the total energy release according to Bortfeld is given by 
\begin{equation}
T(\ell) =- \frac{{\dd} }{{\dd} \ell}\Psi(\ell). 
\label{energyrelease}
\end{equation}

If we return to identity \eqref{BP1d}, with $\e_0$  as the initial energy per unit sequential proton track then we note that, in the language of Bortfeld, $E(\ell) = \mathbb{E}_{(\epsilon_0, 0)}[\e_\ell | \ell <\Lambda]$. Moreover, with  $\eta$ as the average number of sequential proton tracks that are released at the boundary, it follows that the total particle fluence satisfies $\Phi(\ell)= \eta \mathbb{P}_{(\epsilon_0, 0)}(\ell< \Lambda)$.  Hence,
\begin{align}
\mathtt{B}(\ell) & = - \eta\frac{\dd}{\dd \ell}\mathbb{E}_{(\epsilon_0, 0)}[\e_\ell; \ell<\Lambda]= 
- \frac{\dd}{\dd \ell}\Big(\mathbb{E}_{(\epsilon_0, 0)}[\e_\ell | \ell<\Lambda]\mathbb{P}_{(\epsilon_0, 0)}( \ell<\Lambda)\eta\Big)  = T(\ell).
\label{comparebortfeld2}
\end{align}
 
 As such, the SDE model \eqref{SDE3} (as a dimensionally collapsed version of \eqref{VSDE}) offers a consistent entry point to the reasoning of Bortfeld \cite{bortfeld}.
 Beyond this  starting point, we are unable to make further comparison. The reason for this is that from this departure point,  Bortfeld makes a number of empirically  informed approximations to the energy released, $E$, and flux, $\Phi$, and subsequently derives a parametric family of curves that have sufficient numbers of parameters that allow it to be calibrated to experimental data.

It is nonetheless worth mentioning that Bortfeld's derivation of this parametric family is not based on a coherent underlying physical model, but rather a patchwork of localised observations about individual qualitative features that are generally known for the physical phenomenon of the Bragg curve. Ultimately, Bortfeld's parametric family of curves has proved to be a highly successful tool for medical physicists and clinicians because such curves are {\it a priori} pre-shaped to match the qualitative structure of a typical Bragg curve and are coded by a sufficient number of parameters that can be fine tuned to allow effective fitting to real data.

\section{Diffusive approximation}\label{sec:diffusive}
In this section, we discuss the case where \eqref{VSDE} is approximated by a purely diffusive SDE. 
To this end, let us suppose that the driving SDE for  $Y = ((\e_\ell, r_\ell, \Omega_\ell), \ell<\Lambda)$ takes the form 

\begin{equation}
\boxed{
\begin{array}{rl}
\displaystyle \e_\ell &= \e_0-\displaystyle\int_0^\ell \tilde{\varsigma} (Y_l)\dd l\\
&\\
\displaystyle r_\ell& = r_0 +\displaystyle\int_0^\ell\Omega_l\dd l \\
&\\
\displaystyle \Omega_\ell& = \Omega_0-\displaystyle\int_0^\ell m(Y_l)^2\Omega_l\dd l + \int_0^\ell m(Y_l)\Omega_l\wedge\dd B_l,
\end{array}
\label{SDE}
}
\end{equation}
for $\ell< \Lambda =  \inf\{s >0: \e_s =0 \text{ or } r_s\not\in D\}$,
 and where $B= (B_\ell, \ell\geq0)$ is an independent Brownian motion on $\mathbb{R}^3$. In order to ensure that \eqref{SDE} really does approximate \eqref{VSDE}, we choose $\tilde\varsigma$ such that the mean energy deposition rate is the same. That is, we set
\[
\tilde\varsigma(x) =\varsigma(x) +\epsilon \sigma(x)\int_{(0,1)}u\pi(x; \mathbb{S}_2, \dd u), \qquad x \in \mathcal C.
\] 
We also note that unlike \eqref{VSDE}, the SDE \eqref{SDE} does not experience jumps. In particular, equation for $\Omega_\ell$ is that of a state-dependent diffusion on $\mathbb{S}_2$. 
This implies that the energy of a sequential proton track cannot jump to zero. In this case, this can be achieved via an instantaneous (state-dependent) killing rate say $\tilde{q}(x)\in B^+(\mathcal{C})$, which take to be 
\[
  \tilde{q}(x) =\sigma(x) \pi (x; \mathbb{S}_2, \{1\}),\qquad x\in \mathcal{C},
\]
in order to (again) match the behaviour of \eqref{VSDE}.

We note that even if we do not take $\tilde\varsigma$ and $\tilde q$ as above, then the analogous result to Theorem \ref{density} still holds under quite general assumptions. The following result is proved in Appendix B.

\begin{theorem}\label{diffusivethrm}
Suppose that $\tilde\varsigma,  m\in B^+(\mathcal{C})$ are uniformly Lipschitz, uniformly bounded away from zero and are elements of $C^\infty_b(\mathcal C)$. Further assume that $\tilde\varsigma(\epsilon, r, \varphi, \vartheta)^{1/2} + \partial^2_{\tilde\vartheta\vartheta}(\varsigma(\epsilon, r, \varphi, \vartheta)^{1/2})$ is bounded away from $0$ and $\infty$, where we have written $(\varphi, \vartheta)$ in place of $\omega$ using polar coordinates.  Then \eqref{SDE} has a strong solution and its resolvent has a smooth density with respect to ${\rm Leb}_\mathcal{E}\times{\rm Leb}_D\times\Sigma$.
 
\end{theorem}

As before we will use $\mathbb{P} = (\mathbb{P}_x, x\in\mathcal{C})$ to denote the law of the solution to \eqref{SDE}. 
In that case the energy deposited still follows the definition in \eqref{Uopp}, albeit now we have,
\begin{align}
\mathtt{U}[f](x) &= -\mathbb{E}_x\left[\int_0^\Lambda f(Y_{\ell-})\dd \epsilon_\ell\right]\notag\\
&= 
\mathbb{E}_x\left[\int_0^{\kappa_D} f(Y_{\ell}){\rm e}^{-\int_0^\ell \tilde{q}(Y_l)\dd l}
\tilde\varsigma(Y_\ell)\dd \ell\right] + 
\mathbb{E}_x\left[\int_0^{\kappa_D} f(Y_{\ell})\epsilon_\ell \tilde{q}(Y_\ell){\rm e}^{-\int_0^\ell \tilde{q}(Y_l)\dd l}
\dd \ell\right]
\notag\\
&
=\mathbb{E}_x\left[\int_0^\Lambda f(Y_{\ell})\left(\tilde\varsigma(Y_\ell) +\epsilon_\ell \tilde{q}(Y_\ell)\right)\dd \ell\right], \qquad x\in  \mathcal{C},  f\in B^+( \mathcal{C}),
\label{Uoppdiff}
\end{align}
where, again, $\kappa_D = \inf\{\ell>0: r_\ell\not\in D\}$.
Following the calculations of Sections  \ref{AC} and \ref{BM}, under the assumptions of Theorem \ref{diffusivethrm}, we can continue to develop similar calculations, albeit  we now identify 
\begin{equation}
-\left\langle\frac{\dd \e}{\dd \ell}\right\rangle = \tilde\varsigma(\epsilon, r,\Omega) + \epsilon \tilde{q}(\epsilon, r,\Omega), \qquad (\epsilon, r,\Omega)\in \mathcal{C}.
\label{BB2}
\end{equation}
In particular the definition of the Bragg manifold and the Bragg surface remain well-defined.

\section{Monte Carlo simulations}\label{simulations}
We now discuss simulations in 3D of the proton beam using the SDE model \eqref{VSDE}. {\it Most importantly we use the existence of the formula \eqref{BS} to generate the Bragg surface. Moreover, we show that the restriction of the Bragg surface to 1D reduces to a Bragg curve with characteristic Bragg peak.}  In simulating \eqref{VSDE} we found that, without quantifying the relevant rates in proportions that are found in nature,  the output of the SDE is not naturally representative of what would normally be seen in a proton beam. That is to say, the balance of parameter choices that occur naturally are somewhat particular to the SDE model in terms of producing the Bragg peak behaviour and it is essential that the SDE is calibrated to the physical system that it models. In future work, we will demonstrate how to calibrate the SDE to natural settings with precision. However, for the purpose of this illustrative example we make informed choices of parameters, based roughly on what is known of the approximate orders of magnitude of the quantities in the SDE that appear in proton beams,  to demonstrate that the SDE responds accordingly delivering a Bragg peak, and indeed a Bragg surface as our theory predicts.

\subsection{Illustrative rate functions}\label{specifics}

We first discuss the various rates involved in \eqref{VSDE}, drawing on insights from physics literature, primarily \cite{NZ} and \cite{Harald}.  For many of the cross sections, we have been unable to find closed form expressions, in which case we take inspiration from empirical data. {\it  We therefore emphasise that the functions we work with in this section are purely illustrative for the purpose of demonstrating the qualitative nature of Monte Carlo solutions to our SDE. The reader will note that one may easily adapt the analytical structure of the cross sections as needed. In this sense our SDE model is robust.}
 In future work we will undertake a comparative study with clinical and phantom data to demonstrate comparatively how the Bragg surface can be calibrated in a consistent way with current uses of the Bragg curves, where numerically derived cross sections can be used.

 We note that whilst our analysis in previous sections are taken with $r\in D\subset \mathbb{R}^3$ and, accordingly, $\Omega\in\mathbb{S}_2$, the theory we have developed works equally well in a planar setting where $D\subset\mathbb{R}^2$ and $\mathbb{S}_2$ is replaced by the unit circle, $\mathbb{S}_1$.

\smallskip

\noindent {\bf Bethe--Bloch formula and inelastic proton-electron collisions.} Recall from our previous discussion that the Bethe--Bloch formula \eqref{realBB} can be written in terms of the components of our SDE, as given by \eqref{BB1}.
According to \cite{NZ}, the Bragg--Kleeman rule (see \cite{BK}), which was originally derived for alpha particles, suggests
$$
-\left\langle { \frac{\dd \epsilon}{\dd\ell}} \right\rangle
= -\frac{\epsilon^{1-p}}{\alpha p},
$$
where $\alpha = \alpha(r)$ is a material-dependent constant and the exponent $p = p(\epsilon)$ is a constant that takes into account the dependence of the proton’s energy or velocity. Values of $\alpha$ and $p$ may be obtained by fitting to either ranges or stopping power data from measurements or theory.

 Again, purely for illustrative purposes, we shall take $\alpha =  0.022$ and  $p = 1.77$; cf. \cite{bortfeld} for supporting rationale. 

\smallskip

\noindent {\bf Large elastic Coulomb scattering cross sections.}
In the study of elastic Coulomb scattering, Rutherford's theory of single scattering provides valuable insights. 
It describes the behaviour of particles as they interact through Coulomb forces, resulting in deviations in their trajectories. Central to this theory is the concept of the differential cross-section, denoted as $\dd \sigma_{\rm e}$, which characterises the likelihood of particles scattering into a specific differential solid angle $ \dd\Omega$ for $\Omega\in \mathbb{S}_2$.

In terms of the notation introduced in the paragraph preceding \eqref{Npi}, the expression for the differential cross section, as derived from Rutherford's theory and elaborated in \cite{Particle}, is given by 
\begin{align*}
\sigma_{\rm e}(x)\pi_{\rm e}(x; \dd \Omega') = \bigg(\frac{Z_eZ'_e}{4\epsilon}\bigg)^2\frac{1}{\sin^4(\varphi/2)}\dd \Omega',
\end{align*}
for $x = (\epsilon, r, \Omega)\in\mathcal{C}$ and where $\Omega'\in\mathbb{S}_2$ is parameterised by $(\theta, \varphi)\in [0,2\pi)\times [0,\pi]$. 
The other parameters are as follows: $Z'_e$ represents the charge of the incident particle, $Z_e$ the charge of the scattering nucleus and $\epsilon$ is the energy of the particle. Note that the angular dependence is dictated by the term $1/\sin^4(\varphi/2)$ which is $O(1/\varphi^4)$ for small values of $\varphi$. This  indicates a preference for forward scattering at very small angles. The lack of azimuthal dependence signifies a uniform dependence on the azimuthal angle of scatter. 
By integrating, we can obtain the total elastic cross section,
\begin{align*}
    \sigma_{\rm e}(x)=\bigg(\frac{Z_eZ'_e}{4\epsilon}\bigg)^2\int_0^{2\pi}\dd\theta\int_0^\pi \dd\varphi\frac{1}{\sin^4(\varphi/2)}\sin\varphi= \infty.
\end{align*}
Informally, this indicates that particles undergo scattering events with vanishingly small scatter angles at an infinite rate---a process which is not amenable to simulation.
In our model, the elastic cross section will only capture the jumps with large scatter angle; all small scatter angle displacements are treated by the  diffusive angular term of the SDE.  Hence, we cut the angular integration off at some $\iota>0$. Then, since
$
    \int_\iota^\pi \dd\varphi \sin^{-4}(\varphi/2)\sin \varphi<\infty,
$
we observe that $\sigma_{\rm e}\propto {1}/{\epsilon^2}$. For simulations we will take ${\color{black}\iota = 0.01}$ and, by introducing a multiplicative constant as a free parameter, specify
\begin{align}\label{eq: elastic rate}
{\color{black}    \sigma_{\rm e}(x)=\frac{1}{100\epsilon^2},\qquad x\in{\mathcal C}.}
\end{align}
This implies that
\begin{equation}\label{eq: elastic dis}
    \pi_{\rm e}(x; \dd \theta,\dd \varphi ) = \frac{1}{C}\frac{\cos(\varphi/2)}{\sin^3(\varphi/2)} \,\dd \varphi\,\dd \theta, \qquad (\theta, \varphi)\in [0,2\pi)\times [\iota,\pi],
\end{equation}

\begin{wrapfigure}{l}{6cm}
  \begin{center}
  \hspace{-10pt}
\includegraphics[width=5.5cm, height = 2.5cm]{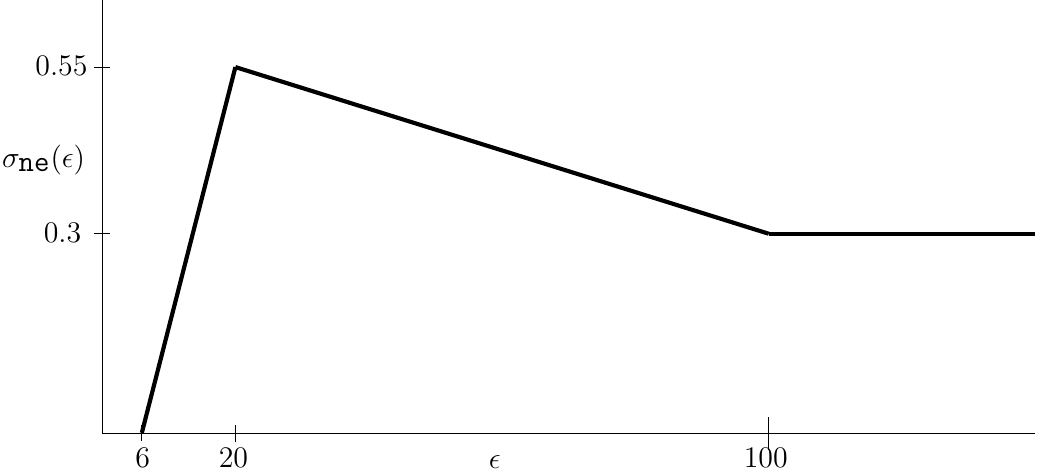}
\caption{The numerical functional we use for $10^3\sigma_{\rm ne}$ used in the simulations.}
    \label{sigma_ne}
      \vspace{-20pt}
\end{center}
\end{wrapfigure} 
\noindent where, with an abuse of notation, we have written $\pi_{\rm e}(x; \dd \theta,\dd \varphi )$ in place of $\pi_{\rm e}(x; \dd \Omega )$, $C = 2\pi[\sin^{-2}(\iota/2)-1]$
and we have used that $\sin\varphi = 2\sin(\varphi/2)\cos(\varphi/2)$.
\smallskip

\noindent {\bf Small elastic Coulomb scattering and inelastic scattering:} Again, for simplicity and illustrative reasons we shall assume that $m(x)= 1/100\epsilon^2$ for $x\in\mathcal{C}$, motivated by the fact that diffusive scattering is modelling elastic interactions with small scatter angles. 

\smallskip

\noindent {\bf Non-elastic scattering.}
Certain qualitative patterns are evident in ~\cite[Figure 8]{NZ}. 
Notably, the non-elastic cross section may be taken purely as a function of energy and maintains a constant value at higher energy levels, before exhibiting an increase as energy decreases, followed by a subsequent decline. 
To capture this behaviour, we adopt a crude approximation for the non-elastic scattering cross section based on ~\cite[Figure 8]{NZ}, so that $10^3\sigma_{\rm ne}(\epsilon)$ is  numerically recovered from the graphic in Figure \ref{sigma_ne}.

For the outgoing direction after a non-elastic jump, recall that we parametrise $\Omega \in \mathbb S_2$ by the polar coordinates $(\theta, \varphi)\in [0,2\pi)\times [0,\pi]$. Once again, abusing notation and writing $\pi(x; \dd \theta',\dd \varphi', \dd u)$ in place of $\pi(x; \dd \Omega', \dd u)$, for $\theta'\in[0,2\pi)$, $\varphi'\in[0,\pi]$, $u\in(0,1)$, we set
\begin{align}\label{beta}
\pi_{\rm ne}(x; \dd\theta', \dd\varphi', \dd u) = 
  \frac{1}{2\pi^2}\frac{(\frac{\varphi'}{\pi})^{\alpha_\varphi-1}(1-\frac{\varphi'}{\pi})^{\beta_\varphi-1}}{B(\alpha_\varphi,\beta_\varphi)}\frac{(\frac{\theta'}{2\pi})^{\alpha_\theta -1}(1-\frac{\theta'}{2\pi})^{\beta_\theta-1}}{B(\alpha_\theta,\beta_\theta)}\,
\dd \theta'\,
\dd \varphi'\,\dd u
\end{align}
where $B(\alpha,\beta)={\Gamma(\alpha)\Gamma(\beta)}/{\Gamma(\alpha+\beta)}$ and $\Gamma$ denotes the Gamma function. 
We choose $\alpha_i$ and $\beta_i$ for $i\in\{\varphi,\theta\}$ in the following manner,
\begin{align*}
    (\alpha_\varphi, \beta_\varphi)&=\begin{cases}
        (2,\frac{\pi}{\varphi})\quad &\text{if }0\leq\varphi\leq \frac{\pi}{2}\\
        (\frac{\pi}{\pi-\varphi},2)\quad &\text{if }\frac{\pi}{2}<\varphi\leq \pi,
    \end{cases}\\
    (\alpha_\theta, \beta_\theta)&=\begin{cases}
        (2,\frac{2\pi}{\varphi})\quad &\text{if }0\leq\theta\leq \pi\\
        (\frac{2\pi}{2\pi-\varphi},2)\quad &\text{if }\pi<\theta< 2\pi.
    \end{cases}
\end{align*}

We have assumed that $q(x): =\sigma(x) \pi (x; \mathbb{S}_2, \{1\}) = 0$ but brought to an end all tracks for which the energy drops below 0.01.
\smallskip

\subsection{Simulation of full 3D proton beam}
Finally, we work with the full 3D SDE \eqref{VSDE} to simulate a proton beam consisting of $10^6$ sequential proton tracks.
The C++ code for replicating this simulation is available at \url{https://github.com/JereKoskela/proton-beam-sde}. We take the phase space to be $\mathcal C = [0, 100] \times D \times \mathbb S_2$, where the spatial domain $D \approx [0, 80] \times [-5,5] \times [-5,5]$ is determined dynamically by the simulation code to store every proton track until absorption, and discretised into voxels with side length 0.02.
While the cross sections in the simulations are indicative, for the purpose of comparison the numerical values of energy correspond approximately to megaelectron volts, while the scale of the spatial variables can be thought of as approximate millimetres.

Again representing points on $\mathbb S_2$ in polar coordinates $(\theta, \varphi) \in [0, 2\pi) \times [0, \pi]$, we take the initial measure $\eta$ to be 
\[\eta = \sum_{i = 1}^{10^6} \delta_{(\epsilon_i, r_i, (0, \pi/2))},\]
where each $r_i$ is independently and uniformly sampled from the set $\{(0, y, z) : 
|y + z|^2\leq 1\}$, representing a circular nozzle of radius 1 located at the origin and pointing along the $x$-axis.
Each $\epsilon_i$ is randomised uniformly between 95 and 100.

Schematically, the simulation of a proton trajectory with illustrative rate functions as described above proceeds as follows.
The rates of elastic and non-elastic Coulomb scattering are both bounded above, and are straightforward to simulate via Poisson thinning.
Specifically, given a current proton configuration $(\epsilon_0, r_0, (\theta_0, \varphi_0))$, we simulate the track length until the next candidate scattering event with rate
\begin{equation*}
\sigma^* = \sigma_{\rm e}(\epsilon_a) + \sigma_{\rm ne}^*,
\end{equation*}
where $\sigma_{\rm ne}^* = \max_{\epsilon \in \mathcal{E}}\{ \sigma_{\rm ne}(\epsilon)\}$, $\epsilon_a = 0.01$ is the absorption energy, and we take advantage of the fact that energy decreases along trajectories and $\sigma_{\rm e}$ is a decreasing function so that $\sigma^*$ is an upper bound on the jump rate $\sigma(\epsilon)$.
Let the resulting realisation of track length be denoted by $L$.
Then the intervening track is subdivided into $n_L := \lceil L / \delta \rceil$ intervals of length $L / n_L$, where $\delta > 0$ is a user-specified discretisation parameter which we set to 0.01 by default.
We then build the discretised proton track $\{ ( \epsilon_i, r_i, (\theta_i, \varphi_i) ) \}_{i = 1}^{n_L}$ by successively updating
\begin{align*}
\epsilon_i &= \epsilon_{i - 1} - \frac{\epsilon_{i - 1}^{1-p}}{\alpha p} \frac{ L }{ n_L },\\
\Big(\begin{array}{c}
\theta_i \\
\varphi_i
\end{array}\Big) &= \Big(\begin{array}{c}
\theta_{i-1} \\
\varphi_{i-1}
\end{array}\Big) + \Delta B_{ m(\epsilon_{i - 1}) L / n_L },\\
r_i &= r_{i - 1} + \left(\begin{array}{c}
\sin\theta_i\cos\varphi_i\\  
\sin\theta_i\sin\varphi_i\\  
\cos\theta_i
\end{array}
\right) \frac{ L }{ n_L },
\end{align*}
where $\Delta B_\ell$ is an increment of spherical Brownian motion run for time step $\ell$, which we generate using the exact algorithm of \cite{Mijatovic}, and we have written $m(\epsilon)$ as shorthand for $m(x)$ since the only dependency on the state $x = (\epsilon, r, \Omega)$ is on the energy variable $\epsilon$.
Once the track length reaches the candidate scattering event, it is accepted with probability $\sigma(\epsilon_{ n_L } ) / \sigma^*$.
Rejected scattering events are disregarded, while accepted ones are resolved by determining whether the event is elastic or non-elastic, and updating the energy and direction of the proton using the appropriate distribution: $\pi_{\rm e}$ or $\pi_{\rm ne}$.
This process is repeated until the proton is absorbed.

%


\smallskip

In order to produce a realisation of the Bragg surface, we only consider the $x$-$y$ plane. We project the function $r\mapsto\mathtt{B}(r)$ to the $x$-$y$ plane by integrating out the $z$-variable. That is to say, if $r = (x,y,z)\in\mathbb{R}^3$, then we plot $(x,y)\mapsto\int\mathtt{B}(x,y, z)\dd z$. We also consider a 2D slide through the 3D surface  $r\mapsto\mathtt{B}(r)$ by simply setting $z=0$. That is to say, we consider $(x,y)\mapsto \mathtt{B}(x,y,0)$. Finally, we also project onto one-dimension to recover the characteristic Bragg curve by considering $x\mapsto\int\int\mathtt{B}(x,y,z)\dd y\dd z$.
The resulting energy deposition profiles are shown in Figure \ref{3Surface}, which took 75 minutes to generate on a 2017 Intel i7-6500U core.
\begin{figure}[h!]
  \centering
    \includegraphics[width=.3\textwidth, height = .3\textwidth]{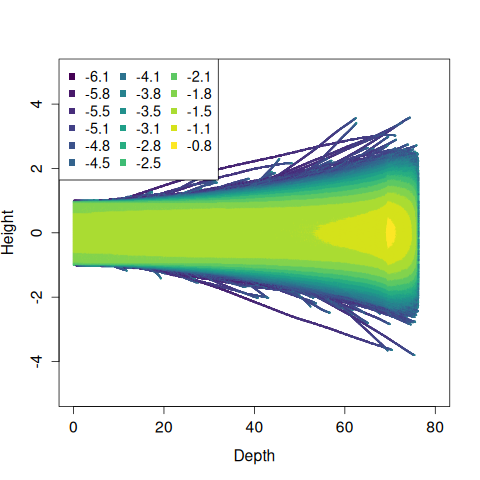}
      \includegraphics[width=.3\textwidth, height = .3\textwidth]{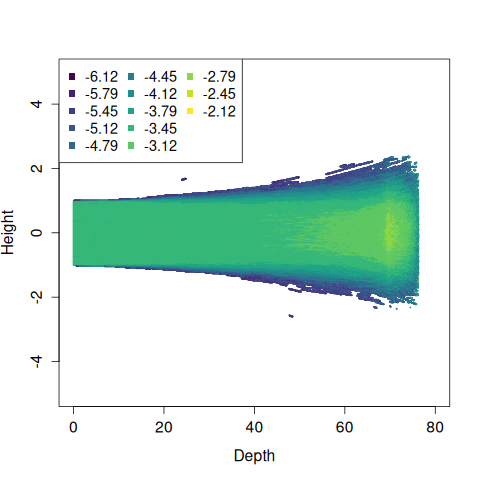}
  \includegraphics[width=.3\textwidth, height = .3\textwidth]{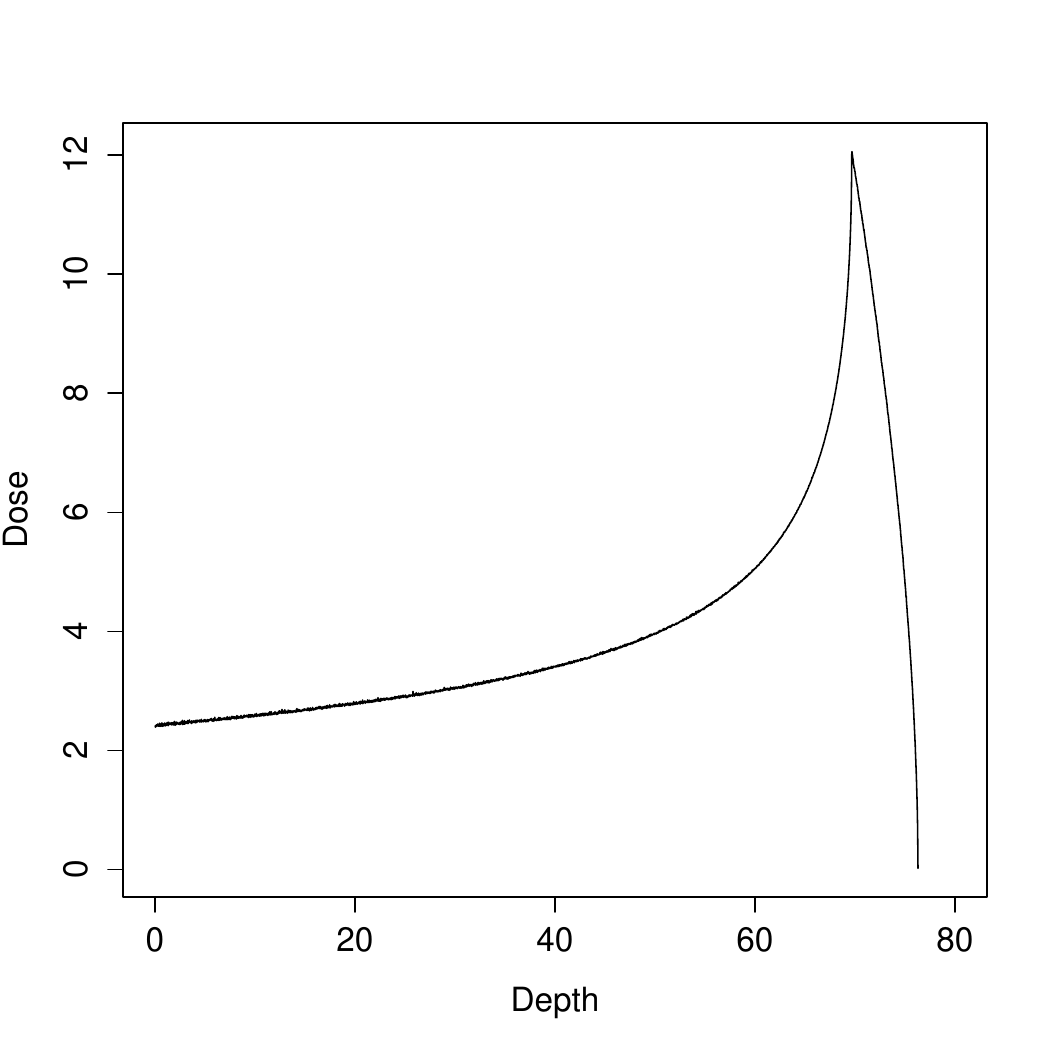}
  \captionof{figure}{(left) Realisation of $10^6$ proton paths in 3D projected in to 2D with the height map, with height on the log scale in base 10, representing the resulting value of the Bragg surface $(x,y)\mapsto\int\mathtt{B}(x,y,z)\dd z$. (centre) a 2D slice through the 3D height map of the stopping power $(x,y)\mapsto\mathtt{B}(x,y, 0)$ for the same simulated protons beam, with again with height on the log scale (base 10), added according to stopping power.  (right) The projection of the Bragg surface  onto the $x$-axis (this time on the linear scale), $x\mapsto \int\int\mathtt{B}(x,y, z)\dd y\dd z$,  giving a classical rate of energy deposition resembling a Bragg peak.}
  \label{3Surface}
\end{figure}


\section{Conclusion}
We have provided a mathematical framework in the form of an SDE solution, which respects the physics of proton transport in proton beam therapy. Our presentation is focused purely on the robustness of the mathematical structure in that it offers a  conceptual definition of the natural higher dimensional analogue of the Bragg curve, which is  meaningful in terms of the underlying physics as well as existing as a well defined mathematical object. In sense, our presentation is as much about how to extend the Bethe--Bloch formula into a fully probabilistic model for the effects of subatomic stopping power.

The SDE approach can be used as a sound mathematical basis for future work on calibration to nuclear data (see \cite{ICRU63}), and hence provides the connection with high-fidelity simulations in e.g.\ Geant4. Indeed, in future work, we will demonstrate how this can be done, resulting in a fast high-fidelity model of the proton beam.

Perhaps most importantly, our approach offers the mathematical basis from which one can numerically solve inverse problems. One such example involves treatment planning. Here, one must back-engineer the initial configuration of the proton beam so as to provide a desired dose delivery profile across a series of voxelated regions overlaying a tumour, given information (e.g.\ from a CT scan) of the patient's anatomy. A second such inverse problem concerns recent work on recording secondary gamma radiation emitted from the patient during treatment, and using it as a conditioning data set to produce uncertainty quantification for dose verification, see e.g.\ \cite{Bayesian}. We will also demonstrate in future work how the SDE approach can be incorporated into these contexts to produce completely new inversion algorithms.

\section*{Appendix A: Spherical Brownian motion in 3 dimensions}\label{SBM3}

Spherical Brownian motion falls under the heading of diffusions on manifolds which, in principle, is a rather heavy-duty topic. That said, we only need to work with the setting of Brownian motion on $\mathbb{S}_2$, which means we may take advantage of various specialisms. Below we consider  two representations.
Before doing so, we also remark that Monte Carlo simulation of spherical Brownian motion is also well understood, see e.g.\ \cite{Mijatovic}.

\subsection*{Spherical coordinate representation}

With $B = (B_\ell, \ell \ge 0)$ as an $\mathbb{S}_2$-valued Brownian motion, we have the option to identify, on the one hand, 
\begin{equation}
B_\ell=
\left(\begin{array}{c}
\sin\vartheta_{\ell}\cos\varphi_{\ell}\\  
\sin\vartheta_{\ell}\sin\varphi_{\ell}\\  
\cos\vartheta_{\ell}
\end{array}
\right), \qquad\ell\geq0,
\label{Bhasdensity}
\end{equation}
where $(\varphi_\ell, \vartheta_\ell)\in [0,\pi]\times [0,2\pi)$ are the spherical coordinate angles, that is, the polar (or co-latitudinal) and azimuthal angles. According to \cite{IM}, the pair $( \varphi, \vartheta)$ necessarily satisfies the coupled SDE
\begin{align*}
&\dd \varphi_\ell = \dd W^{(1)}_\ell + \frac{1}{2}\cot\varphi_\ell \dd \ell, \notag\\
&
\dd \vartheta_\ell = \frac{1}{\sin \varphi_\ell}\dd W^{(2)}_\ell,
\end{align*}
where $W^{(1)},W^{(2)}$ are independent standard Brownian motions  on $\mathbb R$. The first of these two equations for $\varphi$ is nothing more than the Brownian Taboo process on $[0,\pi]$. More precisely, it lives on $(0,\pi)$ and never touches the boundary points $\{0, \pi\}$; see \cite{Knight}. The equation for $\vartheta$ is that of a time-changed Brownian motion. 

For completeness, we mention that Brownian motion on $\mathbb{S}_1$ is more easily captured as 
\begin{equation}
B_\ell=
\left(\begin{array}{c}
\cos W_\ell\\  
\sin W_\ell
\end{array}
\right), \qquad\ell\geq0,
\label{Bplanar}
\end{equation}
where $W = (W_\ell, \ell\geq0)$ is a standard Brownian motion on $\mathbb R$.

\subsection*{Cartesian coordinate representation}

An alternative representation for spherical Brownian motion is as a solution to the SDE
\begin{equation}
\dd B_\ell = B_\ell \wedge\dd \widetilde{W}_\ell - B_\ell\dd \ell
=
A(B_\ell)
\dd\widetilde{W}_\ell - B_\ell\dd \ell,
\label{BSDE}
\end{equation}
where $(\widetilde{W}_\ell, \ell\geq0)$ is a standard Brownian motion on $\mathbb{R}^3$ and 
\begin{equation}
A(x) = \left(\begin{array}{ccc}
0 & -x_3 & x_2 \\
x_3 & 0 & -x_1\\
-x_2 &x_1 &0
\end{array}\right), \qquad x= (x_1,x_2,x_3)\in\mathbb{R}^3.
\label{A}
\end{equation}

If we assume that $(m_\ell, \ell\geq0)$ is non-negative, adapted to the natural filtration of $B$, and that it is uniformly bounded away from 0 and $\infty$, then we can define the Markovian time change
\[
\tau_\ell = \inf\Big\{l>0: \int_0^l m^{-2}_u\dd u> \ell\Big\}, \qquad \ell\geq0.
\]
The time changed process $(B_{\tau_\ell}, \ell\geq0)$ is   equivalent to the solution $(B^\tau_\ell, \ell\geq0)$ to the SDE
\begin{align}
\dd B^\tau_\ell &= m_\ell B_\ell \wedge\dd \widetilde{W}_\ell - m^2_\ell B_\ell\dd \ell \notag
\\
&=\left(\begin{array}{ccc}
0 & - m_\ell B^{(3)}_\ell & m_\ell B^{(2)}_\ell \\
m_\ell B^{(3)}_\ell & 0 & -m_\ell B^{(1)}_\ell\\
-m_\ell B^{(2)}_\ell &m_\ell B^{(1)}_\ell &0
\end{array}\right)
\dd\widetilde{W}_\ell - m_\ell^2 B_\ell\dd \ell.
\label{BtauSDE}
\end{align}


\section*{Appendix B: Proof of Theorem \ref{density} and \ref{diffusivethrm}}
It turns out that it is more convenient to undertake the proof of Theorem \ref{diffusivethrm} first, as it feeds the proof of Theorem \ref{density}.

\begin{proof}[Proof of Theorem \ref{diffusivethrm}]
For ease, we break this proof into several steps. In step 1, we will show that the SDE \eqref{SDE} admits a unique strong solution. In step 2 we introduce an auxiliary SDE via a time change of \eqref{SDE} that will be useful in step 3, where we show that the solution to the auxiliary SDE admits a smooth density which is absolutely continuous with respect to ${\rm Leb}_\mathcal{E}\times{\rm Leb}_D\times \Sigma$. In addition, we show that this implies that the same holds for \eqref{SDE}. In step 4, we show that the existence of such a density still holds when we restrict sequential proton tracks to $\ell < \Lambda$ by introducing absorption at the boundary of $D$. Finally, in step 5, we show that this result can also be extended to the case when we introduce absorption at a bounded rate inside $D$. 

\smallskip

\noindent{\it Step 1.} Consider the SDE  in integral form, indexed by track length, $Z_\ell = (\epsilon_\ell,  r_\ell, \Omega_\ell)\in \mathcal{E}\times\mathbb{R}^3 \times \mathbb{S}_2$ given by
\begin{align}
Z_\ell& = 
z_0
+
\left(
\begin{array}{l}
-\int_0^\ell\varsigma (Z_{{l}})\dd l \\ \\
\int_0^\ell\Omega_{l}\dd l \\  \\
-\int_0^\ell m(Z_l)^2\Omega_l\dd l + \int_0^\ell m(Z_{l})\Omega_l\wedge\dd B_{l} 
\end{array}\right),
\label{VSDE2}
\end{align}
where $z_0 = (\epsilon_0, r_0,\Omega_0)$ and $\varsigma, m$ are subject to the same conditions as in Theorem \ref{diffusivethrm}.
The SDE \eqref{VSDE2} simply models what is known in the literature as a state-dependent {\it kinetic diffusion} with the addition of a depleting  energy term. Under the assumptions of Theorem \ref{diffusivethrm}, a unique strong solution to this SDE is a well known result, see for example Section 5.2 of \cite{Oksendal}.

Henceforth we will write $\mathbf{P} = (\mathbf{P}_{\epsilon, r,\Omega}, (\epsilon, r, \Omega)\in \mathcal{E}\times\mathbb{R}^3\times \mathbb{S}_2)$ for its probabilities and recall e.g. from Chapter 7 of \cite{Oksendal} that $(Z, \mathbf{P})$ is a strong Markov process with the Feller property (see Corollary 19.27 of \cite{Schilling}). 

\smallskip

\noindent {\it Step 2.} We now introduce an auxiliary SDE obtained as a time change of \eqref{VSDE2}.
Taking account of the discussion in  Appendix A above, in particular \eqref{BtauSDE},  the system \eqref{VSDE2} can equivalently be studied by considering the integral form of the SDE for $W_\ell = (\epsilon_\ell, r_\ell, \varphi_\ell, \vartheta_\ell)\in \mathcal{E}\times \mathbb{R}^3\times  [0,\pi]\times [0,2\pi)$
\begin{align}
W_\ell& = 
w_0
+
\left(
\begin{array}{l}
-\int_0^\ell\varsigma (W_{{l}})\dd l \\ \\
\int_0^\ell \Omega_{l}\dd l\\   \\
\int_0^\ell m(W_l)^2(\cot\varphi_l)\dd l/2 + \int_0^\ell m(W_l)\dd \hat{W}^{(1)}_{l} 
\\ \\
\int_0^\ell {m(W_l)}({\sin \varphi_l})^{-1} \dd \hat{W}^{(2)}_l \end{array}\right),
\label{VSDE3}
\end{align}
where $\hat{W}^{(1)}, \hat{W}^{(2)}$ are independent standard Brownian motions and $w_0 = (\epsilon_0, r_0, \varphi_0, \vartheta_0)$.

Noting that $(\epsilon_\ell, \ell\geq0)$ is strictly decreasing, 
we can reduce this system even further by changing from indexing by track length to indexing by dispelled energy $E = \epsilon_0 - \epsilon$. Indeed, we note that 
\begin{equation}
E_\ell = \epsilon_0 - \epsilon_\ell = \int_0^\ell \varsigma(\epsilon_0- E_l, r_l, \Omega_l)\dd l.
\label{energychange}
\end{equation}
Now define for each $0\leq \varepsilon\leq \epsilon_0$, 
\[
E^{-1}_\varepsilon = \inf\{\ell>0: E_\ell >\varepsilon\}.
\]
Basic calculus tells us that 
\begin{equation}
\frac{\dd E^{-1}_\varepsilon}{\dd \varepsilon}  =\left(\left.\frac{\dd E_\ell}{\dd \ell}\right|_{\ell = E^{-1}_\varepsilon}\right)^{-1} = \frac{1}{\varsigma (\varepsilon, r_{E^{-1}_\varepsilon}, \Omega_{E^{-1}_\varepsilon})}.
\label{E-1differential}
\end{equation}

Next, define $R_\varepsilon = r_{E^{-1}_\varepsilon}$ and ${\mathcal S}_{\varepsilon} = \Omega_{E^{-1}_\varepsilon}$, $\Phi_\varepsilon = \varphi_{E^{-1}_\varepsilon}$ and $\Theta_\varepsilon = \vartheta_{E^{-1}_\varepsilon}$ for $0\leq \varepsilon\leq \epsilon_0$. Taking account of \eqref{E-1differential}, the  re-parameterisation by energy allows us to reduce \eqref{VSDE3} to the system $V_\varepsilon = (R_\varepsilon, \Phi_\varepsilon, \Theta_\varepsilon)$, $0\leq \varepsilon\leq \epsilon_0$, where 
\begin{align}
V_\varepsilon = 
w_0
&+
\left(
\begin{array}{l}
\int_0^{\varepsilon}\varsigma(e, R_e,{\mathcal S}_e)^{-1/2}{\mathcal S}_e\dd e\\   \\
\int_0^{\varepsilon}{(2\varsigma(e, R_e,{\mathcal S}_e))}^{-1}m(V_e)^2\cot\Phi_e\dd e 
+\int_0^{\varepsilon} \varsigma(e, R_e,{\mathcal S}_e)^{-1/2}m(V_e)\dd \tilde{W}^{(1)}_e
\\ \\
\int_0^{\varepsilon}\varsigma(e, R_e,{\mathcal S}_e)^{-1/2}{m(V_e)}({\sin \Phi_e})^{-1} \dd \tilde{W}^{(2)}_e \end{array}\right)
\label{VSDE4}
\end{align}
where $\tilde{W}^{(1)}, \tilde{W}^{(2)}$ are independent standard Brownian motions. For convenience, we will denote the law of $({\color{blue}{\mathcal S}}, \Phi, \Theta)$ associated to that of $(Z, \mathbf{P})$ by $\hat{\mathbf{P}} = (\hat{\mathbf{P}}_w:w = (r, \varphi,\vartheta)\in \mathbb{R}^3\times[0,\pi]\times[0,2\pi) )$.

 We note that for a bounded continuous $g:\mathbb{R}^3\times[0,\pi]\times[0,2\pi)$ and $w_0 = (r_0, \varphi_0,\vartheta_0)\in \mathbb{R}^3\times[0,\pi]\times[0,2\pi)$, with $(\Omega_\ell^{(1)}, \Omega_\ell^{(2)}, \Omega_\ell^{(3)})$ denoting the components of $\Omega_\ell$ in $\mathbb R^3$, we have
  \begin{align}
 \hat{\mathbf{E}}_{w_0}\left[g(R_\varepsilon, \Phi_\varepsilon, \Theta_\varepsilon)\right] &= 
  \hat{\mathbf{E}}_{w_0}\left[g(r_{E^{-1}_\varepsilon},  \varphi_{E^{-1}_\varepsilon} , \vartheta_{E^{-1}_\varepsilon})\right] \notag\\
  &=
    \hat{\mathbf{E}}_{w_0}\left[g(r_{E^{-1}_\varepsilon},   \arctan(\Omega^{(2)}_{E^{-1}_\varepsilon}/ \Omega^{(1)}_{E^{-1}_\varepsilon}), \arccos(\Omega^{(3)}_{E^{-1}_\varepsilon}))\right].
 \end{align}
 
On account of the fact that $\varsigma\in B^+(\mathcal{C})$ is uniformly bounded away from 0 and $\infty$, the variable change \eqref{energychange} is continuous in $\varepsilon$ and hence it is easy to see that 
$
\lim_{\varepsilon\to0} \hat{\mathbf{E}}_{w_0}\left[g(R_\varepsilon, \Phi_\varepsilon, \Theta_\varepsilon)\right] = g(w_0).
$
Moreover,  provided $\varphi_0\not\in\{0,\pi\}$, the Feller property of $(Z,\mathbf{P})$ implies that 
$
\lim_{w_0'\to w_0} \hat{\mathbf{E}}_{w'_0}\left[g(R_\varepsilon, \Phi_\varepsilon, \Theta_\varepsilon)\right] =\hat{\mathbf{E}}_{w_0}\left[g(R_\varepsilon, \Phi_\varepsilon, \Theta_\varepsilon)\right]. 
$ 
To accommodate for the setting that $\varphi_0\in\{0,\pi\}$, we note that the parameterisation  of the spherical (or equivalently of the Cartesian) coordinates that we are using is arbitrary and hence a different parameterisation would avoid the correspondence that  $\varphi_0\in\{0,\pi\}$. In that case, the above argument gives the continuity of $w_0\mapsto \hat{\mathbf{E}}_{w_0}\left[g(R_\varepsilon, \Phi_\varepsilon, \Theta_\varepsilon)\right]$.
In conclusion, the Feller property of $(Z,\mathbf{P})$ implies the Feller property of $(V, \hat{\mathbf{P}})$.
 
\smallskip

\noindent {\it Step 3.} We now show that \eqref{VSDE2} has a potential density which is absolutely continuous with respect to ${\rm Leb}_\mathcal{E}\times{\rm Leb}_{\mathbb R^3}\times \Sigma$, where we recall, for a set $S$, ${\rm Leb}_S$ is  Lebesgue measure on $S$, and $\Sigma$ is the surface measure on $\mathbb{S}_2$. In order to do so, we appeal to Malliavin calculus to first show that $V_\varepsilon$ has a density with respect to ${\rm Leb}_{\mathbb{R}^3}\times {\rm Leb}_{[0,\pi]}\times {\rm Leb}_{[0,2\pi)}$ for each $0\leq \varepsilon\leq \epsilon_0$.  By simple a transformation of variables, this will then give us that the  pair $(R_\varepsilon, {\mathcal S}_\varepsilon)$ has a density with respect to ${\rm Leb}_{\mathbb{R}^3}\times \Sigma$ for each $0\leq \varepsilon\leq \epsilon_0$. 

Our approach will be to apply H\"ormander's theorem for time-inhomogenous SDEs, see \cite{Hormander, Cattiaux}. To this end, we write \eqref{VSDE4} in polar coordinates: 
\begin{align}
\dd R_\varepsilon &= \varsigma(\varepsilon, V_\varepsilon)^{-1/2}
\begin{pmatrix} 
\sin\Theta_\varepsilon \cos\Phi_\varepsilon\\
\sin\Theta_\varepsilon \sin\Phi_\varepsilon\\
\cos\Theta_\varepsilon
\end{pmatrix}
\dd \varepsilon,\label{dR}\\
\dd \Phi_\varepsilon &= (2\varsigma(\varepsilon, V_\varepsilon))^{-1}m(V_\varepsilon)^2 \cot\Phi_\varepsilon\dd \varepsilon + \varsigma(\varepsilon, V_\varepsilon)^{-1/2}m(V_e)\dd \tilde{W}^{(1)}_\varepsilon, \label{dPhi}\\
\dd \Theta_\varepsilon &= \varsigma(\varepsilon, V_\varepsilon)^{-1/2}m(V_\varepsilon)(\sin\Phi_\varepsilon)^{-1}\dd\tilde{W}^{(2)}_\varepsilon. \label{dTheta}
\end{align}
In order to apply H\"ormander's theorem, we next write the SDE in the following form
\begin{align*}
  {\rm d}Y_t &= \sum_{j = 1}^m X_j(t, Y_t) {\rm d}W_t^j + X_0(t, Y_t) {\rm d}t, \qquad t \ge 0,
\end{align*}
where the measurable flows $t \mapsto X_j(t, \cdot)$ are assumed to be $C^\infty$ vector fields, such that for any multi-index $\alpha$, $\partial_x^\alpha X_j(t, \cdot)$ are bounded on any compact time-space subset. Comparing this form with \eqref{dR}--\eqref{dTheta}, we have 
\begin{align}
X_0 &= \varsigma(\epsilon, r, \varphi, \vartheta)^{-1/2}\sin\vartheta\cos\varphi \,\partial_{r_1} + \varsigma(\epsilon, r, \varphi, \vartheta)^{-1/2}\sin\vartheta\sin\varphi \,\partial_{r_2} \notag\\
&\hspace{2cm}+ \varsigma(\epsilon, r, \varphi, \vartheta)^{-1/2}\cos\vartheta \,\partial_{r_3} + \frac12\varsigma(\epsilon, r, \varphi, \vartheta)^{-1}m(r, \varphi, \vartheta)^2\cot\varphi \,\partial_\vartheta, \label{X0}\\
X_1 &= \varsigma(\epsilon, r, \varphi, \vartheta)^{-1/2}m(r, \varphi, \vartheta)\, \partial_\varphi, \label{X1}\\
X_2 &= \varsigma(\epsilon, r, \varphi, \vartheta)^{-1/2}m(r, \varphi, \vartheta)(\sin\varphi)^{-1}\, \partial_\vartheta, \label{X2}
\end{align}
where we have slightly abused notation and written $\varsigma(\epsilon, r, \varphi, \vartheta)$ in place of $\varsigma(\epsilon, r, s)$. Note that the assumptions of Theorem \ref{density} ensure that $X_0$, $X_1$ and $X_2$ satisfy the required smoothness and boundedness conditions. With this notation in hand, according to \cite{Cattiaux}, $V_\epsilon$ has a smooth density if 
\[
  \text{dim span Lie}\{\partial_\epsilon + X_0, X_1, X_2\} = 6,
\]
that is, the vector fields $\partial_\epsilon + X_0, X_1, X_2$ and their Lie brackets of all orders span the tangent space at every point $(\epsilon, r, \varphi, \vartheta)$ in the six-dimensional space $[0, \epsilon_0] \times \mathbb R^3 \times [0, \pi] \times [0, 2\pi)$. The Lie bracket $[X, Y]$ of two vector fields $X$ and $Y$ is defined to be $[X, Y] = XY -  YX$.

Before showing that this condition holds, we make two observations to simplify our calculations. First, define
\begin{align}
Y_0 &= \sin\vartheta\cos\varphi \,\partial_{r_1} + \sin\vartheta\sin\varphi \,\partial_{r_2}+ \cos\vartheta \,\partial_{r_3} + \frac12\varsigma(\epsilon, r, \varphi, \vartheta)^{-1/2}m(r, \varphi, \vartheta)^2\cot\varphi \,\partial_\vartheta, \label{Y0}\\
Y_1 &= \partial_\varphi, \label{Y1}\\
Y_2 &= \partial_\vartheta. \label{Y2}
\end{align}
Then, it is easy to check that the brackets generated by $\varsigma(\epsilon, r, \varphi, \vartheta)^{1/2}\partial_\epsilon + Y_0$, $Y_1$ and $Y_2$ are the same as those generated by $\partial_\epsilon + X_0$, $X_1$ and $X_2$ and hence, it is equivalent to show that 
\[
  \text{dim span Lie}\{\varsigma(\epsilon, r, \varphi, \vartheta)^{1/2}\partial_\epsilon + Y_0, Y_1, Y_2\} = 6.
\]
Next, we note that $Y_1$ and $Y_2$ already give us access to the $\partial_\varphi$ and $\partial_\vartheta$ directions, respectively. Thus, it is sufficient to show that
\[
  \text{dim span Lie}\{\varsigma(\epsilon, r, \varphi, \vartheta)^{1/2}\partial_\epsilon + Z_0, Y_1, Y_2\} = 6,
\]
where
\begin{align}
Z_0 = \sin\vartheta\cos\varphi \,\partial_{r_1} + \sin\vartheta\sin\varphi \,\partial_{r_2}+ \cos\vartheta \,\partial_{r_3}. \label{Z0}
\end{align}

To this end, we have
\begin{align}
  Z_1 &=[Y_1, \varsigma(\epsilon, r, \varphi, \vartheta)^{1/2}\partial_\epsilon + Z_0] \notag \\
  &= \partial_\varphi(\varsigma(\epsilon, r, \varphi, \vartheta)^{1/2})\partial_\epsilon - \sin\vartheta\sin\varphi \,\partial_{r_1} + \sin\vartheta\cos\varphi\,\partial_{r_2}, \label{Z1}\\
  Z_{11} &= [Y_1, [Y_1, \varsigma(\epsilon, r, \varphi, \vartheta)^{1/2}\partial_\epsilon + Z_0] ]\notag \\
  &= \partial^2_{\varphi \varphi}(\varsigma(\epsilon, r, \varphi, \vartheta)^{1/2})\partial_\epsilon - \sin\vartheta\cos\varphi \,\partial_{r_1} - \sin\vartheta\sin\varphi\,\partial_{r_2}, \label{Z11}\\
  Z_{22} &= [Y_2, [Y_2, \varsigma(\epsilon, r, \varphi, \vartheta)^{1/2}\partial_\epsilon + Z_0]] \notag \\
  &= \partial^2_{\vartheta\vartheta}(\varsigma(\epsilon, r, \varphi, \vartheta)^{1/2})\partial_\epsilon - \sin\vartheta\cos\varphi\, \partial_{r_1} - \sin\vartheta\sin\varphi \, \partial_{r_2}-\cos\vartheta\, \partial_{r_3}.  \label{Z22}
\end{align}

First note that
\[
   \chi_0 := \varsigma(\epsilon, r, \varphi, \vartheta)^{1/2}\partial_\epsilon + Z_0 + Z_{22} = (\varsigma(\epsilon, r, \varphi, \vartheta)^{1/2} + \partial^2_{\vartheta\vartheta}(\varsigma(\epsilon, r, \varphi, \vartheta)^{1/2}))\partial_\epsilon,
\]
which gives us access to the $\partial_\epsilon$ direction. Thus we can use this to ``remove'' the vectors in the $\partial_\epsilon$ direction from the $Z_i$  since $\varsigma(\epsilon, r, \varphi, \vartheta)^{1/2} + \partial^2_{\vartheta\vartheta}(\varsigma(\epsilon, r, \varphi, \vartheta)^{1/2})$ is bounded away from $0$ and $\infty$. To avoid introducing further notation, we still label \eqref{Z0}--\eqref{Z22} the same once the corresponding $\partial_\epsilon$s are removed. Then, we have
\[
  \chi_1 := \sin\varphi Z_1 + \cos\varphi Z_{11} = -\sin\vartheta \partial_{r_1}
\]
and
\[
  \chi_2 := \cos\varphi Z_1 - \sin\varphi Z_{11} = \sin\vartheta \partial_{r_2}.
\]
Then finally, 
\begin{align*}
 \varsigma(\epsilon, r, \varphi, \vartheta)^{1/2}\partial_\epsilon + Z_0 
 - \frac{ \varsigma(\epsilon, r, \varphi, \vartheta)^{1/2}}{\varsigma(\epsilon, r, \varphi, \vartheta)^{1/2} + \partial^2_{\vartheta\vartheta}(\varsigma(\epsilon, r, \varphi, \vartheta)^{1/2})}\chi_0+ \cos\varphi\chi_1 - \sin\varphi\chi_2 = \cos\vartheta \partial_{r_3}.
\end{align*}
Hence, we now have access to all directions and thus it follows that $(R_\epsilon, \Phi_\epsilon, \Theta_\epsilon)$ has a density with respect to ${\rm Leb_{\mathbb R^3}} \times {\rm Leb_{[0, \pi]}} \times {\rm Leb_{[0, 2\pi)}}$. A change of variables argument then yields the same result for $(R_\epsilon, \mathcal S_\epsilon)$ with respect to ${\rm Leb_{\mathbb R^3}} \times \Sigma$.
 
\smallskip

With the associated density of $V_\varepsilon$, $0\leq \varepsilon\leq \epsilon_0$,  in hand we note that, with $\zeta = \inf\{\ell>0: \epsilon_\ell = 0\}$, for bounded measurable $f, g, h$ on $[0,\epsilon_0]$, $\mathbb{R}^3$, $\mathbb{S}_2$ respectively, we have 
\begin{align}
&\mathbf{E}_{(\epsilon_0, r_0,\Omega_0)}\left[\int_0^{\zeta}f(\epsilon_\ell)g(r_\ell)h(\Omega_\ell)\dd \ell\right] \notag\\
&= 
\mathbf{E}_{(\epsilon_0, r_0,\Omega_0)}\left[\int_0^{\epsilon_0}
\frac{f(\varepsilon) g(R_\varepsilon)h({\mathcal S}_\varepsilon)}{\varsigma(\varepsilon, R_\varepsilon,{\mathcal S}_\varepsilon)}\dd \varepsilon\right]\notag\\
&=\int_0^{\epsilon_0}f(\varepsilon)
\mathbf{E}_{(\epsilon_0, r_0,\Omega_0)}\left[\frac{
g(R_\varepsilon)h({\mathcal S}_\varepsilon)}{\varsigma(\varepsilon, R_\varepsilon,{\mathcal S}_\varepsilon)}
\right]\dd\varepsilon\notag\\
&= \int_0^{\epsilon_0}\int_{\mathbb{R}^3}\int_{\mathbb{S}_2}f(\varepsilon)
g(R )h({\color{blue}{\mathcal S}} )
\frac
{\mathtt{v}(\epsilon_0, r_0,\Omega_0; \varepsilon, R, {\color{blue}{\mathcal S}})}{\varsigma(\varepsilon, R ,{\color{blue}{\mathcal S}} )}
\,\dd\varepsilon\,\dd R\,\dd {\color{blue}{\mathcal S}}
\label{mathttv}
\end{align}
where we have invoked the change of variables given by \eqref{E-1differential}  and  $\mathtt{v}$ is the density of  the pair $(R,{\color{blue}{\mathcal S}})$ under $\mathbf{P}_{(\epsilon_0, r_0,\Omega_0)}$. In conclusion, we see that the resolvent of $((\epsilon_\ell, r_\ell, \Omega_\ell), 0\leq \ell\leq \Lambda)$ is absolutely continuous with respect to ${\rm Leb}_{\mathcal{E}}\times{\rm Leb}_{\mathbb{R}^3}\times \Sigma$. This concludes Step 2.


 \smallskip

\noindent {\it Step 4.} We now  want to restrict the law of triplet $(\epsilon, r,\Omega)$ to that of the law with the physical position $(r_\ell, \ell\geq0)$ killed on first exiting the physical domain $D$. To this end, let  $\kappa_D = \inf\{\ell>0: r_\ell\not\in D\}$. Then by the principle of path counting, we also have for bounded measurable $f,g,h$,
recalling $z_0 = (\epsilon_0, r_0,\Omega_0)$,
\begin{align*}
&\mathbf{E}_{z_0}\left[\int_0^{\zeta\wedge \kappa_D}f(\epsilon_\ell)g(r_\ell)h(\Omega_\ell)\dd \ell\right]\notag\\
&=\int_0^{\epsilon_0}\int_{\mathbb{R}^3}\int_{\mathbb{S}_2}
f(\epsilon)g(r)h(\Omega)\frac{\mathtt{v}(z_0; \epsilon, r, \Omega)}{\varsigma(\epsilon, r, \Omega)} \,\dd \epsilon\,\dd r\,\dd \Omega\notag\\
&\hspace{1cm}-
\int_{0}^\infty \dd u
\int_{\mathcal{E}\times\partial D\times \mathbb{S}_2}\mathbf{P}_{z_0} (\kappa_D\in \dd u,\, \epsilon_{\kappa_D}\in \dd \epsilon',\, r_{\kappa_D}\in \dd r', \, \Omega_{\kappa_D}\in \dd \Omega', \, u< \zeta)\notag\\
&\hspace{4cm}
\int_0^{\epsilon_0}\int_{\mathbb{R}^3}\int_{\mathbb{S}_2}
f(\epsilon)g(r)h(\Omega)\mathbf{1}_{(\epsilon\leq \epsilon')}\frac{\mathtt{v}(\epsilon', r',\Omega'; \epsilon, r, \Omega)}{\varsigma(\epsilon, r, \Omega)}\, \dd \epsilon\,\dd r\,\dd \Omega,
\end{align*}
where  we have used  that the solution to \eqref{VSDE2} has the strong Markov property.

This tells us that the resolvent of $(\epsilon_\ell, r_\ell, \Omega_\ell)$, $\ell\geq0$, with killing on exiting $D$, has a joint density with respect to ${\rm Leb}_\mathcal{E}\times{\rm Leb}_{D}\times \Sigma$. We also note for later that, again  by the principle of path counting, there  also exists a distributional density $\mathtt{v}^D(z_0; \varepsilon, r, \upsilon)$ of the process $(R,{\mathcal S})$ killed when $R$ first exits $D$.
 
\smallskip

\noindent {\it Step 5.} Suppose now we want to consider the law of $Z$ up to killing the process with an instantaneous rate $q(Z_\ell)\geq0$. That is, we are interested in the probabilities 
$(\mathbf{P}^{D, q}_{(\epsilon, r,\Omega)},  \epsilon\in\mathcal{E}, r\in D, \Omega\in\mathbb{S}_2)$, where
\begin{equation}
\mathbf{E}^{D,q}_{(\epsilon_0, r_0,\Omega_0)}\left[f(\epsilon_\ell)g(r_\ell)h(\Omega_\ell) \right]=
\mathbf{E}^D_{(\epsilon_0, r_0,\Omega_0)}\left[{\rm e}^{-\int_0^\ell q(Z_l)\dd l}f(\epsilon_\ell)g(r_\ell)h(\Omega_\ell) \right],
\label{of the form}
\end{equation}
for bounded and measurable $f, g, h$. We note that  with $z_0 = (\epsilon_0, r_0, \Omega_0)$, we have
\begin{align}
&\mathbf{E}^{D,q}_{z_0}\left[\int_0^{\Lambda}f(\epsilon_\ell)g(r_\ell)h(\Omega_\ell)\dd \ell\right]\notag\\
&=\mathbf{E}_{z_0}\left[\int_0^{\epsilon_0}{\rm e}^{-\int_0^{E^{-1}_\varepsilon}q(\epsilon_l, r_l, \Omega_l)\dd l}
\frac{f(\varepsilon) g(R_\varepsilon)h({\mathcal S}_\varepsilon)}{\varsigma(\varepsilon, R_\varepsilon,{\mathcal S}_\varepsilon)}\dd \varepsilon\right]\notag\\
&=\int_0^{\epsilon_0}f(\varepsilon)\mathbf{E}_{z_0}\left[{\rm e}^{-\int_0^{\varepsilon}
\frac{q(e, R_e, {\mathcal S}_e)}{\varsigma(e, R_e, {\mathcal S}_e)}\dd e}
\frac{ g(R_\varepsilon)h({\mathcal S}_\varepsilon)}{\varsigma(\varepsilon, R_\varepsilon,{\mathcal S}_\varepsilon)}\right]\dd \varepsilon.
\label{of the form2}
\end{align}
We can now appeal to  Theorem 1 of \cite{CU} (see also Section 2.2. of \cite{Lupu}) which allows us to write for all bounded measurable $F = F((R_\varepsilon, {\mathcal S}_\varepsilon), 0\leq \varepsilon\leq \epsilon_0)$, bounded measurable $G:\mathcal{C}\mapsto[0,\infty)$ and $z_0 = (\epsilon_0, r_0, \Omega_0)\in \mathcal{C}$
\begin{equation}
\mathbf{E}^D_{z_0}\left[F\cdot G(\varepsilon, R_\varepsilon, {\mathcal S}_\varepsilon)\right] = \int_{ D\times \mathbb{S}_2} \mathbf{E}^D_{\varepsilon, (r_0, \Omega_0)\to (r,\upsilon)}[F]\mathtt{v}^D(z_0;\varepsilon, r, \upsilon)G(\varepsilon, r, \upsilon)\dd r\,\dd \upsilon,
\label{density2} 
\end{equation}
where we understand 
 $\mathbf{P}^D_{\varepsilon, (r_0, \Omega_0)\to (r,\upsilon)}$ is the law of the $\mathbf{P}^D$-bridge of $(R_\varepsilon,{\mathcal S}_\varepsilon)$, $0\leq \varepsilon\leq \epsilon_0$, from $(r_0, \Omega_0)$ to $(r,\upsilon)\in D\times \mathbb{S}_2$ over a time-horizon $\varepsilon$. Noting that the expectation on the  right-hand side of  \eqref{of the form2}  can  be identified in the form \eqref{density2}, this tells us that the resolvent of the process $(\epsilon, r, \Omega)$ under $\mathbf{P}^{D,q}$ has a density, say $\rho^{D,q}(z_0,z)$, $z_0,z\in\mathcal{C}$, with respect to ${\rm Leb}_{\mathcal{E}}\times{\rm Leb}_{D}\times \Sigma$. 
\end{proof}

The proof of Theorem \ref{density} is now essentially a continuation from the previous proof. Accordingly, some of the notation is carried over. 

\begin{proof}[Proof of Theorem \ref{density}]
Without loss of generality, it suffices to consider SDEs of the form 

\begin{align}
Z_\ell&=
z_0
+
\left(
\begin{array}{l}
-\int_0^\ell\varsigma (Z_{{l}})\dd l - \int_0^\ell \int_{(0,1)}
\int_{\mathbb{S}_2} u\e_{{l}-} {N} (Z_{{l}-}; \dd {l}, \dd \up', \dd u)\\ \\
\int_0^\ell\Omega_{l}\dd l \\  \\
-\int_0^\ell m(Z_l)^2\Omega_l\dd l + \int_0^\ell  m(Z_{l})\Omega_l\wedge\dd B_{l} \\
\hspace{2cm}+\int_0^\ell\int_{(0,1)}\int_{\mathbb{S}_2} (\up' -\up_{{l}-}){N} (Z_{{l}-}; \dd {l}, \dd \up', \dd u)
\end{array}\right),
\label{VSDEc}
\end{align}
where for each $x\in \mathcal{C}, \ell\geq0, \up'
\in \mathbb{S}_2$, ${N} (x; \dd \ell, \dd \up', \dd u)$  is an optional random measure  with previsible intensity given by $\sigma(x)\pi (x;  \up', u)\dd \up'\, \dd u\, \dd \ell $,  with $\sigma\geq0$ uniformly bounded and $\pi(x;\cdot)$ a probability density function on $\mathbb{S}_2\times(0,1]$. 
We consider the SDE \eqref{VSDEc} both killed when $(r_\ell, \ell\geq0)$ first exits the domain $D$ as well as killed at instantaneous  rate $q(x) = \sigma(x)\pi(x; \mathbb{S}_2, \{1\})$. This time we use the notation $\mathbb{P}^{D,q}: = (\mathbb{P}^{D,q}_x, x\in\mathcal{C})$ to denote its associated probabilities.

Under the conditions of Theorem 
\ref{density} (in particular because of the boundedness of the rate functions away from 0 and $\infty$), it is a straightforward exercise using the unique strong solution to \eqref{VSDE2}  and the associated strong Markov property to uniquely construct a  strong solution to \eqref{VSDEc} via a standard piecewise path construction (see e.g. the approach in \cite{refracted}), which is also a strong Markov process.
Indeed, recall that $(Z, 
\mathbf{P}^{D,q})$ is a solution to \eqref{VSDE2} with killing on exiting $D$ and at rate $q$. We can construct $Z$ under $\mathbb{P}^{D,q}_x$, $x\in \mathcal{C}$ as follows
\[
Z_\ell = \left\{
\begin{array}{ll}
Z^{(1)}_\ell(x) & \text{ if } 0 = L^{(0)}\leq \ell <L^{(1)}\text{ and }\ell< \Lambda^{(1)},\\
Z^{(1)}_{L^{(1)}-}+ \Delta^{(1)} & \text{ if } \ell = L^{(1)}\text{ and }L^{(1)}< \Lambda^{(1)},\\
 Z^{(2)}_{\ell-L^{(1)}}(Z_{L^{(1)}}) & \text{ if } L^{(1)}<\ell<L^{(2)}\text{ and } \ell-L^{(1)}< \Lambda^{(2)},\\
 Z^{(2)}_{(L^{(2)}-L^{(1)})-}(Z_{L^{(1)}}) + \Delta^{(2)}&\text{ if } \ell= L^{(2)}\text{ and }L^{(2)}-L^{(1)}< \Lambda^{(2)},\\
 \vdots & \vdots\\
 Z^{(n)}_{\ell-L^{(n-1)}}(Z_{L^{(n-1)}}) & \text{ if } L^{(n-1)}<\ell<L^{(n)}\text{ and } \ell-L^{(n-1)}<\Lambda^{(n)},\\
 Z^{(n)}_{(L^{(n)}-L^{(n-1)})-}(Z_{L^{(n-1)}})+ \Delta^{(n)} &\text{ if } \ell= L^{(n)}\text{ and } L^{(n)}-L^{(n-1)}<\Lambda^{(n)},\\
 \vdots & \vdots
\end{array}
\right.
\]
where, for $n\geq1$, we can sequentially define   $Z^{(n)}(x)= (\epsilon^{(n)}(x), r^{(n)}(x), \Omega^{(n)}(x))$ as the strong solution of \eqref{VSDE2}, albeit driven by $B^{(n)}_\ell = B_{\ell + L^{(n-1)}} - B_{L^{(n-1)}}$, with initial state $Z^{(n)}_0(x)=x$. Moreover,  given this value, $L^{(n)}-L^{(n-1)}$ is the time at which $Z^{(n)}(x)$ experiences an event at instantaneous rate $\sigma(x)$,   the vector
\[
\Delta^{(n)} = \left(
\begin{array}{c}-\epsilon^{(n)}_{(L^{(n)}-L^{(n-1)})-}(x) u
\\
 0\\
  \Omega'-\Omega^{(n)}_{(L^{(n)}-L^{(n-1)})-}(x)
\end{array}
\right)
\]
 is distributed according to $\pi(x; \dd \Omega', \dd u)$ and $\Lambda^{(n)} = \inf\{\ell>0 : r^{(n)}_\ell(x)\not\in D\text{ or } \epsilon^{(n)}(x) =0\}$.

To prove that $\mathbb{P}^{D,q}$ has a resolvent density with respect to ${\rm Leb}_\mathcal{E}\times{\rm Leb}_D\times{\Sigma}$, let $(L_n, n\geq0) $ be the jump times (with respect to the parameter $\ell$) of $(Z, \mathbb{P}^{D,q})$ with $L_0 = 0$. 
Suppose that $F:\mathcal{C}\mapsto[0,\infty)$ is bounded and measurable and write for $z_0\in\mathcal{C}$,
\[
\mathtt{R}^{D,q}[F](z_0) = \mathbb{E}^{D,q}_{z_0}\left[\int_0^\Lambda F(\epsilon_{\ell}, r_{\ell}, \Omega_{\ell})\dd \ell\right].
\]
Write $\rho^{D,q+\sigma}(z_0, z)$, $z_0,z\in\mathcal{C}$ for the resolvent density associated with the law  $\mathbf{P}^{D,q+\sigma}_x$ that was defined via the SDE \eqref{VSDE2} killed both on exiting $D$ and with  rate $q + \sigma$ (previously we considered killing at rate $q$ however, the principle of the existence of a resolvent density is still available to us with this slight adjustment). 
Thanks to the strong Markov property, we have,
\begin{align}
\mathtt{R}^{D,q}[F](z_0) 
&= \mathbb{E}^{D,q}_{z_0}\left[\sum_{n = 0}^\infty\mathbf{1}_{(L_n<\Lambda)}
\int_{\mathcal{C}}
F(z)\rho^{D, q + \sigma}(Z_{L_n}; z)\dd z \right]\notag\\
 &=\int_{\mathcal{C}}
F(z) \mathbb{E}^{D,q}_{z_0}\left[\sum_{n = 0}^\infty\mathbf{1}_{(L_n<\Lambda)}
\rho^{D, q + \sigma}(Z_{L_n}; z)
 \right]\dd z,
 \label{inner}
\end{align} 
where the density $\rho^{D, q + \sigma}$ was defined towards the end of the proof of Theorem \ref{density}.
This concludes the current proof.  
\end{proof}

\section*{Acknowledgements} EH, AEK, and SO acknowledge support from the EPSRC grant {\it Mathematical Theory of Radiation Transport: Nuclear Technology Frontiers (MaThRad)}, EP/W026899/2. Moreover, we would like to thank colleagues from the UCLH proton beam facility (in particular Colin Baker) as well as members of the EPSRC MaThRad programme grant team for wider discussion around this text, in particular Maria Laura P\'erez Lara.


\bibliography{references}{}
\bibliographystyle{abbrv}

\end{document}

%% file: preamble.tex
\usepackage{amsthm, enumitem, thmtools}
\usepackage{amssymb}
\usepackage{amsmath,amsfonts}
\usepackage{algorithmic}
\usepackage{hyperref}
\usepackage{array}
\usepackage{subcaption}
\usepackage{textcomp}
\usepackage{stfloats}
\usepackage{url}
\usepackage{verbatim}
\usepackage{graphicx, color}
\usepackage[linesnumbered,ruled]{algorithm2e}

\newcommand{\up}{\Omega}
\newcommand{\e}{{\mathrm \epsilon}}
\newcommand{\dd}{{\mathrm d}}
\usepackage{xcolor}

\definecolor{amethyst}{rgb}{0.6, 0.4, 0.8}
\definecolor{applegreen}{rgb}{0.55, 0.71, 0.0}
\definecolor{aqua}{rgb}{0.0, 1.0, 1.0}
\definecolor{asparagus}{rgb}{0.53, 0.66, 0.42}
\definecolor{armygreen}{rgb}{0.29, 0.33, 0.13}
\definecolor{shitbrown}{rgb}{0.43, 0.21, 0.1}
\definecolor{brightpink}{rgb}{1.0, 0.0, 0.5}
\definecolor{brightube}{rgb}{0.82, 0.62, 0.91}
\definecolor{byzantine}{rgb}{0.74, 0.2, 0.64}

\usepackage{wrapfig}
\definecolor{fig}{RGB}{76,78,92}
\usepackage[labelfont=bf,font={color=fig,small}]{caption}

\usepackage{tikz}
\usetikzlibrary{intersections}
\usetikzlibrary{arrows,snakes,shapes}
\usepackage{pgfplots}
\usepgfplotslibrary{fillbetween}
                
\newcommand{\proton}[1]{%
    \shade[ball color=red] (#1) circle (.25);\draw (#1) node{$+$};
}
\newcommand{\neutron}[1]{%
    \shade[ball color=green] (#1) circle (.25);
}
\newcommand{\electron}[3]{%
    \draw[rotate = #3](0,0) ellipse (#1 and #2)[color=gray];
    \shade[ball color=yellow] (0,#2)[rotate=#3] circle (.1);
}
\newcommand{\nucleus}{%
    \neutron{0.1,0.3}
    \proton{0,0}
    \neutron{0.3,0.2}
    \proton{-0.2,0.1}
    \neutron{-0.1,0.3}
    \proton{0.2,-0.15}
    \neutron{-0.05,-0.12}
    \proton{0.17,0.21}
}

\newcommand{\inelastic}[2]{
  \proton{#1,#2};
  \draw[->,thick,blue](#1+0.5,#2)--(4,-3.7);%
  \draw[->,thick,blue](0,-3)--(-0.3,-4);%
  \shade[ball color=yellow] (-0.3,-4) circle (.1); 
}

\newcommand{\elastic}[2]{
  \proton{#1,#2};
  \draw[->,thick,orange,bend right=90](#1+0.5,#2) to  [out=-30, in=-150] (4,3.);
}

\newcommand{\protoncollision}[3]{
  \proton{#1,#2};
  \draw[->,thick,red](#1+0.5,#2)--(-0.5,0);%
  \draw[snake=coil, line after snake=0pt, segment aspect=0,%
    segment length=5pt,color=red!50!blue] (0,0)-- +(4,2)%
  node[fill=white!70!yellow,draw=red!80!white, below=.01cm,pos=1.]%
  {$\gamma$};%
  \draw[->,thick,red](#1+0.5,#2)--(-0.5,0);%
  \draw[->,thick,red](0.5,0)--(3.7,-1.8);%
  \neutron{4,-2};  
}
\usepackage{xspace}

\usepackage{tikz, pgfplots}
\pgfplotsset{compat=newest}

